\documentclass[a4paper,11pt]{article}
\usepackage{jheppub} % for details on the use of the package, please see the JINST-author-manual
\usepackage{lineno}
\usepackage[T1]{fontenc}
\usepackage{  amsthm, amssymb,latexsym,amsfonts}
\usepackage[final]{pdfpages}
\usepackage{float}
\usepackage{tikz}
\usetikzlibrary{matrix,calc,fit,intersections}
\usepackage{gensymb}
\usepackage{tikz-feynman}
\usepackage{relsize}
\usepackage{graphicx}
\tikzfeynmanset{compat=1.1.0}
\usepackage{circuitikz}
\usepackage{relsize}
 \usepackage{bbold}
\usetikzlibrary{arrows.meta,calc,decorations.markings,math,arrows.meta}
\usepackage{mathtools}
 \usepackage[shortlabels]{enumitem}
\usepackage{afterpage}
\usepackage{enumitem}
\usepackage{tensor}
\usepackage[cal=cm,calscaled=.96]{mathalpha}
\usepackage{hyperref}
 \usepackage{graphicx}
\usepackage{tikz-network}
\usepackage{textgreek,stackengine}
\newcommand\textdot[1]{\stackon[1pt]{\csname text#1\endcsname}{.}}

\usepackage{mathtools}
\DeclarePairedDelimiter\ceil{\lceil}{\rceil}
\DeclarePairedDelimiter\floor{\lfloor}{\rfloor}

\theoremstyle{plain}
\newtheorem{thm}{Theorem}[section] 
\newtheorem{lemma}[thm]{Lemma}

\theoremstyle{definition}

\def\be{\begin{equation}}
\def\ee{\end{equation}}
\def\bea{\begin{eqnarray}}
\def\eea{\end{eqnarray}}

\definecolor{baby}{rgb}{0.96, 0.76, 0.76}

\usepackage{amstext} % for \text macro
\usepackage{array}   % for \newcolumntype macro
\newcolumntype{L}{>{$}l<{$}} % math-mode version of "l" column type
\usepackage{colortbl}
\definecolor{babypink}{rgb}{0.96, 0.76, 0.76}

\newcommand{\mycomment}[1]{}
\title{\boldmath On local fields invariant under the action of topological defects}

\author{Anatoly Konechny}
 \author{and Vasileios Vergioglou}
\affiliation{Department of Mathematics, Heriot-Watt University,\\
 Edinburgh EH14 4AS, U.K.}
\affiliation{Maxwell Institute for Mathematical Sciences,\\
Edinburgh, U.K.}

\emailAdd{A.Konechny@hw.ac.uk, vv2004@hw.ac.uk}

\abstract{In the context of rational conformal field theories (RCFT) we look into the problem of constructing and classifying 
pairs consisting of a local operator and a topological defect which commutes or anticommutes with it. We discuss the bulk and boundary versions of the problem. 
In the latter one considers a conformal boundary condition, a boundary operator on it and a junction with a topological defect. 
In the case of the charge conjugation modular invariant (anti-)commuting configurations in each problem can be obtained when a certain 
restriction on the fusion rules in realised.  We study the corresponding fusion rule problems in detail. While in the bulk case 
it reduces to realising the $a\times b = c$ fusion rule which was studied in \cite{Buican}, in the boundary it leads to a new type of 
problem. We obtain a full solution to this problem for the  $\mathrm{SU(3)}$ WZW theory, thus constructing a  class  of (anti-)commuting 
boundary operators and junctions in that theory, and suggest an approach  to general WZW theories.
}


\definecolor{ashgrey}{rgb}{0.7, 0.75, 0.71}

\DeclareMathOperator{\Hom}{Hom}
\definecolor{mycolor}{HTML}{960018}
\definecolor{coralred}{rgb}{1.0, 0.25, 0.25}
\definecolor{pansypurple}{rgb}{0.47, 0.09, 0.29}
\definecolor{scarlet}{rgb}{1.0, 0.13, 0.0}
\definecolor{baby}{rgb}{0.96, 0.76, 0.76}
\definecolor{aqua}{rgb}{0.5, 1.0, 0.83}
\definecolor{fuc}{rgb}{0.76, 0.33, 0.76}

\newcommand{\dimh}{\operatorname{dim}\operatorname{Hom}}

\newcommand{\aff}{\widehat{\mathfrak{su}}(3)_k}
\newcommand{\kmin}{k_0^{\text{min}}}
\newcommand{\kmax}{k_0^{\text{max}}}
\newcommand{\affg}{\hat{\mathfrak{g}}}

\usepackage{colortbl}
\usepackage{array}
\definecolor{apricot}{rgb}{0.98, 0.81, 0.69}
\definecolor{eslb}{HTML}{48D1CC}

\newcolumntype{a}{>{\columncolor{eslb}}c}

\makeatletter
\newtheorem*{rep@theorem}{\rep@title}
\newcommand{\newreptheorem}[2]{%
\newenvironment{rep#1}[1]{%
 \def\rep@title{#2 \ref{##1}}%
 \begin{rep@theorem}}%
 {\end{rep@theorem}}}
\makeatother

\newreptheorem{theorem}{Theorem}
\newreptheorem{lemma}{Lemma}

\begin{document}

\tikzset{myptr/.style={
        decoration={markings,
            mark= at position 0.7 with {\arrow{#1}} ,
        },
        postaction={decorate}
    }
}

\tikzset{myptr1/.style={
        decoration={markings,
            mark= at position 0.4 with {\arrow{#1}} ,
        },
        postaction={decorate}
    }
}
\tikzset{middlearrow/.style={
        decoration={markings,
            mark= at position 0.5 with {\arrow{#1}} ,
        },
        postaction={decorate}
    }
}

\tikzset{myptr2/.style={
        decoration={markings,
            mark= at position 0.9 with {\arrow{#1}} ,
        },
        postaction={decorate}
    }
}

\tikzfeynmanset{
myblob/.style={
shape=circle,
draw=black}
}

%\newcolumntype{a}{>{\columncolor{babypink}}c}

\maketitle
\flushbottom

\section{Introduction}
\label{sec:intro}

In recent years constraints on RG flows arising from topological defects have been explored in a number of ways (see  e.g. \cite{Graham_Watts,5A,Sakura_etal1,Sakura_etal2,FGS,Gaiotto,Kom_etal,Cordova_etal,Kon1,TanakaNakayama} for  applications in two dimensions). 
In the context of two-dimensional CFTs topological line defects embody the symmetry charges some of which may be non-invertible. 
Given a local field, passing through it a topological defect line in general results in obtaining a configuration with  non-local defect fields.
This is illustrated on Figure \ref{figintro1} where the local and non-local field configurations  are separated out. 
\begin{figure}[H]
    \centering
    \begin{tikzpicture}
    \draw[thick][middlearrow={latex}] (-6,-2.1) -- (-6,1.5);
     \draw[black,fill=black] (-5,-0.25) circle (.5ex);
     \node at (-4.25,-0.2) {$\mathlarger{\mathlarger{=}}$}; 
      \node at (-5,-0.65) {$\phi_{I}$};
       \node at (-6.45,1) {$X_d$};
\draw [thick] (-2,0) arc
    [
        start angle=15,
        end angle=345,
        x radius=0.8cm,
        y radius =0.8cm
    ] ;
    \draw[black,fill=black] (-2.8,-0.25) circle (.5ex);
    \node at (-2.8,-0.65) {$\phi_{I}$};
    \draw[thick] (-2,0) -- (-1.2,0);
    \draw[thick] (-2,-0.4) -- (-1.2,-0.4);
   \draw[thick][middlearrow={latex}] (-1.2,0) -- (-1.2,1.5);
    \draw[thick] (-1.2,-0.4) -- (-1.2,-2.1);
    \node at (-1.6,1) {$X_d$};
    \node at (0.1,-0.2) {$\mathlarger{\mathlarger{=}}\;\mathlarger{\sum_J}Q^0_{IdJ}$}; 
    \draw[black,fill=black] (1.5,-0.25) circle (.5ex);
\node at (1.5,-0.65) {$\phi_{J}$};
\draw[thick][middlearrow={latex}] (2.2,-2.1) -- (2.2,1.5);
\node at (1.8,1) {$X_d$};
\node at (2.9,-0.2) {$\mathlarger{\mathlarger{+}}$}; 
    \node at (4,-0.5) {$\;\mathlarger{\sum_{\substack{J\\e\neq 0\\ \alpha,\beta}}}Q^{\alpha e\beta}_{IdJ}$};
     \draw[black,fill=black] (5.5,-0.2) circle (.5ex);
      \draw[thick][middlearrow={latex}] (5.5,-0.2) -- (7,-0.2);
      \node  at (6.2,0.1) {$X_e$};
      \node at (5.6,-0.8) {$\mu^{e;\beta}_{J}$};
      \node at (7.1,-0.2) (jun) [circle,draw,scale=0.6] {};
      \node at (7.4,-0.2) {$\alpha$};
      \draw[thick][middlearrow={latex}] (jun) -- (7.1,1.5);
      \node at (6.7,1) {$X_d$};
      \draw[thick][middlearrow={latex}] (7.1,-2.1) -- (jun);
      \node at (6.7,-1.5) {$X_d$};
\end{tikzpicture}
    \caption{Sweeping a topological defect across a bulk field in a generic RCFT.}
    \label{figintro1}
\end{figure}
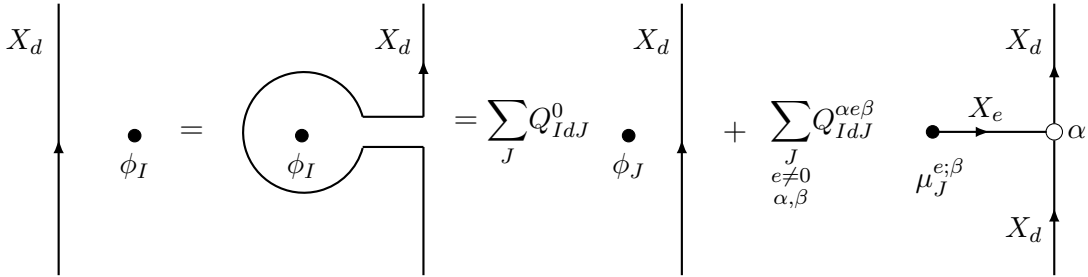

On Figure   \ref{figintro1} the local fields are labeled by the index $I$ while
the sum  in the non-local part runs  over three indices: $e$, $\alpha$ and $\beta$. The index $e$ labels all irreducible topological defects different 
from the identity   (the local field contribution is formally associated with the 
identity defect), $\beta$ labels a basis of the defect fields $ \mu_{I}^{e;\beta}$ residing at the end of $X_{e}$ and the index $\alpha$ 
 labels a basis of topological junctions 
between the defect $X_{e}$ and the original defect $X_{d}$. 
The numbers $Q^{\alpha e\beta}_{IdJ}$ are some model-dependent coefficients.

For rational CFTs  the rule depicted in Figure \ref{figintro1} was studied in \cite{FFRS} and \cite{Runkelcharges} (see  also \cite{KW} for a nice discussion in the context of the Ising model).
The situation when the coefficients $Q^{\alpha e\beta}_{IdJ}$  in the non-local part all vanish  is of particular interest in applications to RG flows. 
In the simplest situation passing the defect results in multiplying the field $\phi_{I}$ by a constant. 
 If this constant is 1 the 
defect $X_d$ commutes with $\phi_I$ and if -1 it anti-commutes. 
Perhaps the simplest example of this situation is the spin-reversal defect in the critical Ising model which commutes with the thermodynamic -- $\epsilon$  
field and anti-commutes with the magnetic field -- $\sigma$. Another pair in the critical Ising model is the Kramers-Wannier defect that anti-commutes with $\epsilon$.
As discussed in \cite{5A} in each case there are implications for the RG flow triggered by 
$\phi_I$ assuming that this field is relevant. Thus, in the case of commutation the corresponding topological defect survives the deformation and must be realised in the IR fixed point (which may be topological if the perturbed theory is gapped). In the anti-commutation case there must be a defect with the same quantum numbers linking the IR fixed points corresponding to the different signs of the coupling for operator $\phi_I$. 

More generally we can envisage a situation when passing through a defect results in a linear combination of local fields with the same quantum numbers, see e.g. the case of perturbed 3-state Potts model discussed in \cite{5A}. In such a case, assuming the fields are relevant, we can consider 
a multi-coupling  deformation  in which these fields  are simultaneously switched on. The corresponding topological defect  should exist between the IR end-points on the trajectories related by the action of the defect. Having this general situation in mind we can say that  the relevant topological defects leave invariant some subspace of local operators. For applications to unitary theories one should add to this the requirement that the corresponding invariant subspace must be Hermitian. 

Although computing the coefficients $Q^{\alpha e\beta}_{IdJ}$ and finding  topological defects that leave invariant some (Hermitian) subspaces of local fields has not been 
systematically addressed, it has been noted that in certain situations the fusion rules give enough information to construct  some  local fields invariant 
under topological defects. To discuss this case in more detail let us specialise to RCFTs with the charge conjugation modular invariant. In this case the bulk fields 
are labeled by a pair $I=(i, \bar i)$ where $i$ labels a chiral irreducible representation with $\bar i$ labelling its conjugate. 
After passing a topological  defect the local part is proportional to the original field while the intermediate defects $e$ are restricted by the fusion rules 
\cite{PZ} as depicted on Figure \ref{figintro2}. 

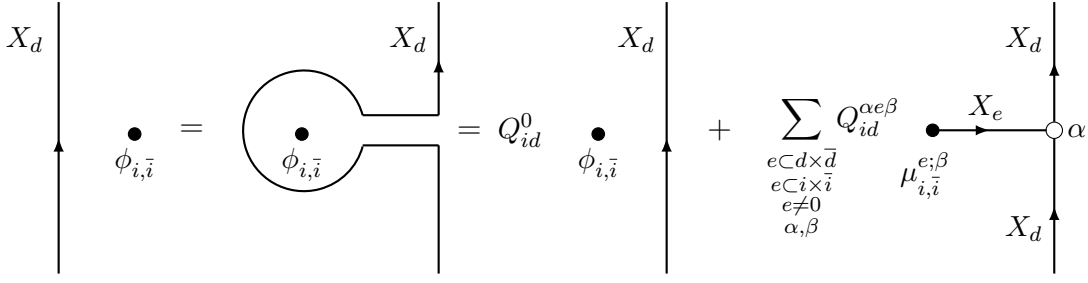
\begin{figure}[H]
    \centering
    \begin{tikzpicture}
    \draw[thick][middlearrow={latex}] (-6,-2.1) -- (-6,1.5);
     \draw[black,fill=black] (-5,-0.25) circle (.5ex);
     \node at (-4.25,-0.2) {$\mathlarger{\mathlarger{=}}$}; 
      \node at (-5,-0.65) {$\phi_{i,\overline{i}}$};
       \node at (-6.45,1) {$X_d$};
\draw [thick] (-2,0) arc
    [
        start angle=15,
        end angle=345,
        x radius=0.8cm,
        y radius =0.8cm
    ] ;
    \draw[black,fill=black] (-2.8,-0.25) circle (.5ex);
    \node at (-2.8,-0.65) {$\phi_{i,\overline{i}}$};
    \draw[thick] (-2,0) -- (-1,0);
    \draw[thick] (-2,-0.4) -- (-1,-0.4);
   \draw[thick][middlearrow={latex}] (-1,0) -- (-1,1.5);
    \draw[thick] (-1,-0.4) -- (-1,-2.1);
    \node at (-1.4,1) {$X_d$};
    \node at (-0.2,-0.2) {$\mathlarger{\mathlarger{=}}\;\;Q^0_{id}$}; 
    \draw[black,fill=black] (1.1,-0.25) circle (.5ex);
\node at (1.1,-0.65) {$\phi_{i,\overline{i}}$};
\draw[thick][middlearrow={latex}] (2,-2.1) -- (2,1.5);
\node at (1.6,1) {$X_d$};
\node at (2.7,-0.2) {$\mathlarger{\mathlarger{+}}$}; 
    \node at (4,-0.7) {$\quad\mathlarger{\sum_{\substack{e\subset d\times \overline{d}\\e\subset i\times \overline{i}\\e\neq 0\\ \alpha,\beta}}}Q^{\alpha e\beta}_{id}$};
     \draw[black,fill=black] (5.5,-0.2) circle (.5ex);
      \draw[thick][middlearrow={latex}] (5.5,-0.2) -- (7,-0.2);
      \node  at (6.2,0.1) {$X_e$};
      \node at (5.4,-0.8) {$\mu^{e;\beta}_{i,\overline{i}}$};
      \node at (7.1,-0.2) (jun) [circle,draw,scale=0.6] {};
      \node at (7.4,-0.2) {$\alpha$};
      \draw[thick][middlearrow={latex}] (jun) -- (7.1,1.5);
      \node at (6.7,1) {$X_d$};
      \draw[thick][middlearrow={latex}] (7.1,-2.1) -- (jun);
      \node at (6.7,-1.5) {$X_d$};
\end{tikzpicture}
    \caption{Sweeping a topological defect across a bulk field in a RCFT with charge conjugation modular invariant.}
    \label{figintro2}
\end{figure}

Using the TFT approach \cite{FRS1,FRS2,FRS3,FRS4,FRS5} we find the following general expression for the coefficients $Q^{\alpha e\beta}_{id}$ 
in terms of the braiding and fusing matrices
\begin{equation}\label{Qgeneral}
    Q^{\alpha e\beta}_{id}=\sum_{j\subset i\times d}\;\sum_{\rho,\mu,\nu,\gamma}\mathrm{R}^{(id)j}_{\mu\rho}\, \mathrm{R}^{(di)j}_{\rho\nu}\, \mathrm{G}^{(di\overline{i})d}_{\gamma j \nu,\beta e \alpha}\,\mathrm{F}^{(di\overline{i})d}_{0,\mu j\gamma} \, .
\end{equation}
Our notational conventions are explained in Appendix \ref{appendix0} where we also give the details of the calculation\footnote{The expression we give generalises the cases considered in  \cite{FFRS}, \cite{Runkelcharges}. While the computation is a straightforward application of the TFT approach we are not aware of a paper where it appears in full generality.}.
The coefficient of the local part can be expressed in terms of the modular $S$-matrix  \cite{PZ}
\be
Q^0_{id} = \frac{S_{id}}{S_{00} {\rm dim}(i) {\rm dim}(d)}
\ee
where ${\rm dim}(i)=\frac{S_{i0}}{S_{00}}$ stands for the quantum dimension of the corresponding representation.

One known set of examples  comes from $A$-series minimal models.
In a minimal model ${\cal M}_{p,p'}$ the local fields as well as  the topological defects are labeled by an entry from the corresponding Kac table: $ (r,s)$, $1\le r\le p'-1$, $1\le s \le p-1$.   The fusion rule of the model is such that all representations of the form $(1,s)$ form a subring as well as those of the form $(r,1)$. 
The only common elements in the two subrings are the identity representation $(1,1)$ and the simple current $(p'-1,1)\equiv (1,p-1)$.
   Choosing $d$ to be given by  $(r,1)$ and $i=\bar i=(1,s)$ with $r\ne p'-1$ and $s\ne p-1$ we obtain that the only common irreducible representation 
   that appears in both $d\times d$ and $i\times i$ is the identity representation. Hence in this case there can be no non-local terms on the right hand side of 
Figure \ref{figintro2} and we have a local field invariant under a topological defect. 

More generally, in an RCFT with a charge conjugation modular invariant we have a local field $\phi_{i,\overline{i}}$ invariant under $X_{d}$ whenever 
\be \label{dimHom}
{\rm dim}\, {\rm Hom}(i\times \bar i\times d \times \bar d, \mathbb{1}) = 1
\ee
where $\mathbb{1}$ stands for the identity representation.
Equation (\ref{dimHom}) is equivalent to saying that the only common irreducible representation among those that appear both in $i\times \bar i$ and in $d\times \bar d$ is 
the identity representation. Another way to read (\ref{dimHom}) is by stating that the fusion product $i\times d$ is itself an irreducible representation. 
If either $i$ or $d$ is a simple current this is always true but as the above example of minimal models shows that's not the only possibility.
The problem of finding such situations when neither $i$ nor $d$ is a simple current was studied in \cite{Buican} in a broader context of fusion categories. 
In particular it was noted in that paper that, using the results of \cite{WZWtadpole}, one can show that  there are no solutions to (\ref{dimHom}) for WZW theories apart from the case when at least one of the representations is a simple current.

% BOUNDARY ---------------------------

While in the bulk case the problem of constructing invariant operators using  fusion rules has been addressed in  \cite{Buican} the analogous boundary case has been not and is the main focus of the present paper. In the boundary case we consider a boundary operator $\Psi_i^\beta$ where the index $i$ labels an irreducible representation of the chiral algebra and $\beta$ is a degeneracy label. The operator $\Psi_i^\beta$ in general is a boundary condition changing operator 
linking two irreducible conformal boundaries: $M_a$ and $M_{b}$. Consider next a topological defect $X_{x}$ that ends topologically on a conformal 
boundary. Without loss of generality we can assume that it links two irreducible conformal boundaries: $M_b$ and $M_{c}$.
Passing the defect-boundary junction through the boundary field results in a sum over similar configurations as depicted on Figure \ref{figintro4}.
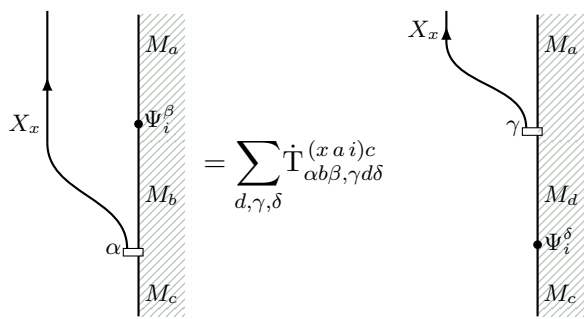
\begin{figure}
    \centering
    \begin{equation*}
     \vcenter{\hbox{\hspace{-10mm}\begin{tikzpicture}[font=\footnotesize,inner sep=2pt]
  \begin{feynman}
   \path[pattern=north east lines,pattern color=ashgrey,very thin] (0.01,4) rectangle (0.6,0);
  \vertex (j1) at (0,0);
  \vertex [above=0.8 of j1] (j2);
  \vertex[ above =0.9 of j1] (j3);
  \vertex [small,dot] [ above=2.5cm of j1] (psi)  {};
  \vertex[right=0.2cm of psi] (psi1) ;
  \vertex[below=0.2cm of psi1] (psi2) [label=\(\,\,\,\Psi_i^\beta\)];
  \vertex [above=3.4cm of j1] (j4) ;
  \vertex[above=3.5cm of j1] (j5)  ;
  \vertex[ left=0.15cm of  j4] (j7);
  \vertex [left=0.15cm of j3] (j8) [label=left:\(\alpha\)]; 
  \vertex [above =4cm of j1] (j6);
  \vertex [below=0.4cm of j6] (p1)  [label=right:\(M_a\)];
  \draw  (0.06,0.8) rectangle (-0.2, 0.9);
    \vertex[below=0.1cm of j6] (m1) ;
    \vertex[above=1.6cm of j1] (m3) [label=right:\(M_b\)];
    \vertex[above=0.9cm of psi] (m2);
    \vertex[left=1.2cm of psi] (x) [label=left:\(X_x\)];
    \vertex[above=0.3 of j1] (jj) [label=right:\(M_c\)];
    \vertex [below=0.25cm of x] (x1);
    \vertex [above=0.25cm of x] (x2);
    \vertex[above=1.5cm of x] (xtop);
    \draw [thick]    (j8) to[out=90,in=-90] (x1);
    \draw [thick] [middlearrow={latex}]   (x1) to (xtop);
   \diagram*{
    (j1) --[thick] (j2),
    (j3) --[thick] (psi),
    (psi) --[thick] (j6)
  };
  \end{feynman}
\end{tikzpicture}}}
~=\mathlarger{\sum}_{d,\gamma,\delta}{\mathrm{\dot{T}}}^{\,(x\,a\,i)c}_{\alpha b\beta,\gamma d\delta}~
\vcenter{\hbox{\hspace{1mm}\begin{tikzpicture}[font=\footnotesize,inner sep=2pt]
  \begin{feynman}
   \path[pattern=north east lines,pattern color=ashgrey,very thin] (0.01,4) rectangle (0.6,0);
  \vertex (j1) at (0,0);
  \vertex [above=2.4 of j1] (j2);
  \vertex[ above =2.5 of j1] (j3);
  \vertex [small,dot] [ above=0.9cm of j1] (psi)  {};
  \vertex[right=0.2cm of psi] (psi1);
  \vertex[below=0.2cm of psi1] (psi2) [label=\(\,\,\,\Psi_i^\delta\)];
  \vertex [above=3.4cm of j1] (j4) ;
  \vertex[above=3.5cm of j1] (j5)  ;
  \vertex[ left=0.15cm of  j4] (j7);
  \vertex [left=0.15cm of j3] (j8) [label=left:\(\gamma\)]; 
  \vertex [above =4cm of j1] (j6);
  \vertex[below=0.4cm of j6] (p2) [label=right:\(M_a\)];
  \draw (0.06,2.4) rectangle (-0.2, 2.5);
    \vertex[below=0.1cm of j6] (m1) ;
    \vertex[above=1.6cm of j1] (m3) [label=right:\(M_d\)];
    \vertex[above=0.9cm of j2] (m2) ;
    \vertex[above=0.5cm of j4] (xx);
    \vertex[left=1.2cm of xx] (x);
    \vertex[below=0.1cm of x] (x100)  [label=left:\(X_x\)];
    \vertex[above=0.3 of j1] (jj) [label=right:\(M_c\)];
    \vertex [below=0.25cm of x] (x1);
    \vertex [above=0.15cm of x] (x2);
    \vertex[left=0.25cm of x2] (x22);
    \vertex[above=0.6cm of j6] (j7);
    \vertex[right=0.25cm of j7] (j77);
    \vertex[left=0.25cm of x2] (x22);
    \vertex[above=0.05cm of x] (x3);
    \vertex[left=0.25cm of x3] (x33);
    \vertex[above=0.1cm of x33] (x333);
    \vertex[above=0.4cm of x2] (x4);
    \draw [thick]    (j8) to[out=90,in=-90] (x1);
    \draw [thick] [middlearrow={latex}]   (x1) to (x2);
   \diagram*{
    (j1) --[thick] (j2),
    (j3) --[thick] (j6)
  };
  \end{feynman}
\end{tikzpicture}}}
\end{equation*}
    \caption{Moving an open topological defect past a boundary field.}
    \label{figintro4}
\end{figure}
The configurations on the right hand side of Figure \ref{figintro4} are weighed by the coefficients 
 which are elements of a set of fusing matrices that we denoted $\dot{\mathrm{T}}$. This fusing matrix and its properties have been discussed in 
 \cite{paper1}. While in general they are distinct from the fusing matrix of the chiral algebra, in the case of the charge conjugation modular invariant, 
 in a suitable basis, they coincide:
 \be
 {\mathrm{\dot{T}}}^{\,(x\,a\,i)c}_{\alpha b\beta,\gamma d\delta} = {\mathrm{F}}^{\,(x\,a\,i)c}_{\alpha b\beta,\gamma d\delta} \, .
 \ee 
The fusing matrix coefficients $\mathrm{F}$ are non-vanishing only if the following restrictions from the fusion rules are satisfied 
\bea \label{fus_conds}
&& x\times a \supset d  \, , \qquad d\times i \supset c  \, , \nonumber \\
&& a\times i \supset b \, , \qquad x \times b \supset c 
\eea  
where the $\supset$ symbol reads as "contains".

Of particular interest is the situation when $a=b=c$ when we have a boundary field and a topological junction with the same irreducible boundary 
condition. The conditions (\ref{fus_conds}) are then equivalent to 
\bea \label{fus_conds2}
&& a \times \bar a \supset i  \, , \qquad a\times \bar a \supset x  \, , \nonumber \\
&& a\times \bar d \supset i \, , \qquad \bar{a} \times  d \supset x \, . 
\eea  
These conditions are always satisfied by $d=a$ which can be separated out and plays the role similar to the local part of the bulk case. 
If there are no terms with $d\ne a$ present in the right hand side of Figure \ref{figintro4} then passing the defect $X_{x}$ through the boundary 
field $\Psi_i^\beta$ leaves behind a superposition of boundary fields (with the same chiral algebra representation)   on the {\it same boundary 
condition}. We can say that we have an invariant subspace of boundary fields. If these fields  are relevant we get a constraint on the 
multi coupling space of boundary RG flows triggered by them. In the simplest situation when we have a single field commuting with a topological junction 
the junction survives the flow and must be present at the IR fixed point boundary condition. If moving the junction $X_{x}$ past the boundary field results in 
multiplying the field by -1, as in the bulk case, the defect $X_{x}$ must have a topological junction between the two IR fixed point for the opposite signs of the coupling.

It follows from (\ref{fus_conds2}) that, in the charge conjugation modular invariant case, we can ensure an invariant one-dimensional subspace if we have a triple of representations: $a,x,i$ such that $x$ and $i$ each appears in $a\times \bar a$ with multiplicity 1 so that 
\be \label{mult1_cond}
{\rm dim}\, {\rm Hom}(a\times \bar a \times i \,  , {\mathbb 1}) = 1 \, , \quad {\rm dim}\, {\rm Hom}(a\times \bar a \times  x\, , {\mathbb 1}) = 1 
\ee
and satisfy 
\be \label{fus_cond3}
{\rm dim}\, {\rm Hom}(a\times \bar a \times i \times  x\, , {\mathbb 1}) = 1 \, .
\ee
The latter condition is equivalent to having $d=a$  as the only common representation that appears both in $a\times  i$ and in $a\times  \bar{x}$. 
We will refer to triples of representations $a,x,i$ satisfying (\ref{mult1_cond}), (\ref{fus_cond3}) as "special triples". 
Although (\ref{fus_cond3}) appears to be similar to (\ref{dimHom}) it is a different problem that in general involves three different representations. 
However, as in the bulk case, it is easy to construct solutions via simple currents in the following way. Suppose $i$ is a simple current and it appears in the fusion $a\times \bar{a}$. Equivalently, the boundary condition $a$ is invariant under the action of $i$, that is $i\times a=a$. In this case  $i$ appears in $a\times \bar a$ with multiplicity 1. Furthermore, since $i\times a \times \bar a  = a \times \bar a$ 
the complete set of representations that appears in $a\times \bar{a}$ is invariant under the action of $i$. As this action preserves the multiplicities 
it suffices to pick for $x$ any representation that appears in $a\times \bar a$ with multiplicity 1. This ensures that (\ref{fus_cond3}) is satisfied. 
We will show in the later sections of the paper that for the $\mathrm{SU(2)}$ and $\mathrm{SU(3)}$ WZW models with the charge conjugation modular invariant 
 this construction with simple currents exhausts all possible special triples of representations. 

It should be emphasised that the bulk equation (\ref{dimHom}) and the boundary equations (\ref{mult1_cond}), (\ref{fus_cond3}) only provide 
sufficient conditions for constructing fields invariant under the action of topological  defects. There may be other solutions for which some coefficients 
vanish even though they are allowed by fusion rules. It would be interesting to find a method for constructing all possible solutions. 

The rest of the paper is organised as follows. In section \ref{sec2} we discuss the specialisation of equations (\ref{mult1_cond}), (\ref{fus_cond3}) to WZW 
theories.  In section \ref{sec3} we prove that the only special triples the in $\mathrm{SU(2)}$ WZW theory come from simple currents. 
Section \ref{sec4} proves the same for $\mathrm{SU(3)}$. Appendix A contains details of the calculation leading to formula (\ref{Qgeneral}). Appendix B contains some technical details 
of the $\mathrm{SU(3)}$ special triples  analysis in section \ref{sec4}.

%%%%%%%%%%

\section{WZW theories and simple currents} \label{sec2}

We consider  WZW models \cite{WZW1,WZW2,WZW6} based on    a simple Lie group $G$. 
As we are only interested in the fusion rule in this section we briefly introduce the basic representation-theoretic aspects of these theories.

Let $\mathfrak{g}$ be the Lie algebra of $G$ and $P_+$  the set of highest weights of the irreducible highest weight representations of $\mathfrak{g}$. The symmetry algebra of such models at level $k$ is $\hat{\mathfrak{g}}_k\times \hat{\mathfrak{g}}_k$ where $\hat{\mathfrak{g}}_k$ denotes the untwisted affine Lie algebra associated with $\mathfrak{g}$. The affine chiral primary fields $\phi_i$ at level $k$ correspond to  irreducible integrable highest weight representations $\{V_{\hat{\lambda}}\,|\,\hat{\lambda}\in P_+^k\}$ of $\hat{\mathfrak{g}}_k$, where we denoted the set of highest weights of such representations by $P_+^k$. For fixed $k$, each affine integrable highest weight $\hat{\lambda}\in P_+^k$ corresponds uniquely to a weight $\lambda\in P_+$ and vice versa.   To simplify notation we will sometimes denote an integrable representation $V_{\hat{\lambda}}$ by just its affine highest weight $\hat{\lambda}$.

\mycomment{
For given $G$ the category  $\mathcal{C}_k(\mathfrak{g})$ generated by the irreducible   integrable $\affg_k$ representations  is a modular tensor category \cite{WZW3,WZW4,WZW5}. The simple objects of this category are $\{V_{\hat{\lambda}}\,|\,\hat{\lambda}\in P_+^k\}$ and the affine tensor product will be denoted by $\times$ as opposed to the tensor product $\otimes$ of $\mathfrak{g}$-representations. To simplify notation we will sometimes denote an integrable representation $V_{\hat{\lambda}}$ by just its affine highest weight $\hat{\lambda}$.
}
The   fusion product of $\affg_k$-representations has the following  direct sum decomposition 
\begin{equation}
    \hat{\lambda}\times \hat{\mu}=\bigoplus_{\hat{\nu}\in P_+^k} N^{(k)\nu}_{\lambda,\mu}\, \hat{\nu}\,.
\end{equation}
where
\begin{equation}
     N^{(k)\nu}_{\lambda,\mu}:= \dimh(V_{\hat{\lambda}}\times V_{\hat{\mu}},V_{\hat{\nu}})\,
\end{equation}
are the multiplicities of the irreducible representations $\hat \nu$. 

We further define the following subsets of interest
\begin{align}
    \mathcal{R}^k_{\lambda,\mu}&:=\{\hat{\nu}\in P_+^k\,|\, N^{(k)\nu}_{\lambda,\mu}\neq 0\} \nonumber\\
    \mathcal{R}^{(1);k}_{\lambda,\mu}&:= \{\hat{\nu}\in P_+^k\,|\, N^{(k)\nu}_{\lambda,\mu}= 1\} \,.
\end{align}

Suppose now  $P_+^k=\{\hat{\lambda}_0,\ldots,\hat{\lambda}_m\}$ where $\hat{\lambda}_0$ denotes the vacuum representation. Among the  chiral primary fields $\{\phi_0,\ldots,\phi_m\}$ we distinguish the subset of simple currents \cite{Schellekens01,Schellekens02,Fuchs01} $\mathcal{J}=\{J_0,\ldots,J_l\}$.
Each simple current is a primary field $J_{n}=\phi_{j_{n}}$ for which the representation $\hat \lambda_{j_{n}}$ satisfies
$N^{(k)i}_{j_n,\overline{j_n}}=\delta_{i,0}$ where $\overline{j}$ denotes the highest weight of the conjugate representation $\overline{V}_{\hat{\lambda}_j}$.
The fusion rules for simple currents with any  primary field 
are of the form 
\begin{equation}
    J_i\times \phi_j = \phi_{j^\prime}\, .
\end{equation}
%A primary field $\phi_j$ is a simple current if and only if $N^{(k)i}_{j,\overline{j}}=\delta_{i,0}$ 

Let us now comment on the rational boundary conditions and the topological defects of such models. As we restrict ourselves to the charge conjugation modular invariant the conformal  boundary conditions and topological defects are labelled by the same set that labels primary fields, namely $P_+^k$. In addition, the tensor product structures  between topological defects and between topological defects and boundary conditions coincide with the fusion product structure between the primary fields.

Equations (\ref{mult1_cond}), (\ref{fus_cond3}) from the introduction in the context of WZW theories  are equivalent to having 
a triple of integrable highest weights $\hat{\lambda},\hat{\mu},\hat{\nu}\in P_+^k$ such that 
\be \label{wzwproperty0}
\hat{\mu},\hat{\nu}\in \mathcal{R}^{(1);k}_{\lambda,\overline{\lambda}} 
\ee and 
\begin{equation}\label{wzwproperty}
    \dimh\left( \hat{\lambda}\times \overline{\hat{\lambda}} \times \hat{\mu}\times \hat{\nu},\hat{\lambda}_0\right)=1\,.
\end{equation}
Here $\hat{\lambda}$ corresponds to a (symmetry-preserving) boundary condition, $\hat{\mu}$ to a topological defect and $\hat{\nu}$ to a boundary field on $\hat{\lambda}$.

As discussed in the introduction the special triples satisfying (\ref{wzwproperty0}), (\ref{wzwproperty}) can be constructed using simple currents. 
Given   a simple current $\hat{\mu}$ and a boundary condition  $\hat{\lambda}$  invariant under fusion with $\hat{\mu}$
%, that is $\hat{\lambda} \times \hat{\mu}= \hat{\lambda}$. 
equation \eqref{wzwproperty} reduces to 
\begin{equation}\label{wzwproperty2}
     \dimh\left( \hat{\lambda}\times \overline{\hat{\lambda}} \times \hat{\nu},\hat{\lambda}_0\right)=1 \Longleftrightarrow  \dimh\left( \hat{\lambda}\times \overline{\hat{\lambda}} , \overline{\hat{\nu}} \right)=1
\end{equation}
that means that for the third representation in the triple we can choose any $\hat{\nu}\in \mathcal{R}^{(1);k}_{\lambda,\overline{\lambda}}$.
In the next two sections we are going to show that in the case of $G=\mathrm{SU(2)}$ and $G=\mathrm{SU(3)}$ such triples exhaust all special triples. % solving 
%(\ref{wzwproperty0}), (\ref{wzwproperty}).
%so that 
%Since $\hat{\nu}\in \mathcal{R}^{(1);k}_{\lambda,\overline{\lambda}}$ we have $\overline{\hat{\nu}}\in \mathcal{R}^{(1);k}_{\overline{\lambda},\lambda}$, but by the symmetry of the fusion product we obtain $\overline{\hat{\nu}}\in \mathcal{R}^{(1);k}_{\lambda,\overline{\lambda}}$. Therefore, \eqref{wzwproperty2} is satisfied for any $\hat{\nu}\in \mathcal{R}^{(1);k}_{\lambda,\overline{\lambda}}$. \textcolor{red}{In the following we will show that triples obtained in this way exhaust all the triples that satisfy \eqref{wzwproperty} in $\mathrm{SU(2)}$ and $\mathrm{SU(3)}$ WZW models.}

\section{A warm up problem: $\mathrm{SU(2)}$ } \label{sec3}

As a warm up  we consider in this section the  problem of finding all special triples  in the $\mathrm{SU(2)}$ WZW models. In this case the set of highest weights of  integrable $\widehat{\mathfrak{su}}(2)_k$ representations is $P_+^k=\{0,1,\ldots,k\}$. We  denote by $(i)$ the representation with highest weight $i\in P_+^k$. The explicit form of the  $\widehat{\mathfrak{su}}(2)_k$ fusion product is
 \begin{equation}\label{su2fus}
        (i)\times(j)=\bigoplus_{p=|i-j|}^{\operatorname{min}(i+j,2k-i-j)}(p)
    \end{equation}
 where $p$ increases in steps of $2$. We also recall that every $\widehat{\mathfrak{su}}(2)_k$-representation is self-conjugate, all fusion multiplicities are at most $1$ and the simple current set is $\mathcal{J}_2=\{0,k\}$.  Therefore, our goal for the $\mathrm{SU(2)}$ case can be reformulated as: find $a,i,j\in P_+^k$ with $i,j\in \mathcal{R}^k_{a,a}$ such that
 \begin{equation}\label{propertysu2}
      \operatorname{dim}\operatorname{Hom}\left((a)\times(a)\times(i)\times(j),(0)\right)=1\,.
 \end{equation}
The solution is provided by the following theorem.
\begin{thm}\label{su2thm}
   Let $a\in P_+^k\setminus \mathcal{J}_2$\footnote{If $(a)$ is any simple current then we have a trivial situation since $\mathcal{R}^k_{a,a}=\{(0)\}$.} and  $i,j\in P_+^k$ such that 
    \begin{equation}\label{con}
    i,j\in \mathcal{R}^k_{a,a}\,.
    \end{equation}
    Then if $i,j\notin \mathcal{J}_2$ the following holds
   \begin{equation}
       \operatorname{dim}\operatorname{Hom}\left((a)\times(a)\times(i)\times(j),(0)\right) >1 \,.
   \end{equation}
\end{thm}
\begin{proof}
Since the fusion product for  $\widehat{\mathfrak{su}}(2)_k$ representations is symmetric, we will assume without loss of generality that $i\geq j$. 
%Under the conditions of the theorem we have 
%    \begin{equation}
 %        \operatorname{dim}\operatorname{Hom}\left((a)\times(a)\times(i)\times(j),(0)\right)\neq 0\,.
  %  \end{equation}
%    Indeed, since every representation is self-conjugate we can rewrite condition \eqref{con} as $a\in \mathcal{R}^k_{i,a}$ and $a\in \mathcal{R}^k_{j,a}$. Then, consider the following two factors in the 4-product of interest
   % \begin{equation}\label{e1}
      %  (i)\times(a)=(i-a)\oplus\cdots\oplus(a)\oplus\cdots\oplus\operatorname{min}(i+a,2k-i-a)
   % \end{equation}
   % and
    %\begin{equation}\label{e2}
     %    (j)\times(a)=(j-a)\oplus\cdots\oplus(a)\oplus\cdots\oplus\operatorname{min}(j+a,2k-j-a).
    %\end{equation}
   % We see that there is a common representation $(a)$ in the r.h.s. of \eqref{e1} and \eqref{e2} which will lead to a copy of $(0)$ in the 4-product $(i)\times(a)\times(j)\times(a)$ coming from the term $(a)\times(a)$.
% Therefore we will prove the equivalent statement 
%\begin{equation}
%     \operatorname{dim}\operatorname{Hom}\left((a)\times(a)\times(i)\times(j),(0)\right) =1 \Longrightarrow \quad \text{at least one of }\, i,j\in \mathcal{J}_2\,.
%\end{equation}
By assumption we have  $0<i,j<k$ 
% arbitrary satisfying \eqref{con} 
and $a\notin\{0,k\}$ such that   %and consider the following two factors in the 4-product of interest
\begin{equation}\label{m}
        (a)\times(a)=(0)\oplus (2)\oplus\cdots(j)\oplus\cdots\oplus (i)\oplus\cdots \oplus \left(\operatorname{min}(2a,2k-2a)\right)
    \end{equation}
    where at least two representations must be present in the direct sum, 
    and
    \begin{equation}\label{ij}
        (i)\times(j)= (i-j)\oplus (i-j+2)\oplus\cdots \oplus \left(\operatorname{min}(i+j,2k-i-j)\right)
    \end{equation}
    The first expansion implies that  $i$ and $j$ are both even. Since $i\ge j$ and $j\ge 2$ the representations $(i-j)$ and $(i-j+2)$ must each appear in the right hand side of (\ref{m}). 
    While $(i-j)$ is also manifestly  present in the second expansion (\ref{ij}) to show that $(i-j+2)$ appears in (\ref{ij}) we need to check that under the assumptions at hand $i-j+2 \le \operatorname{min}(i+j,2k-i-j)$. Since $j\ge 2$ it follows that $i-j+2 \le i + j$. Furthermore, $i-j+2 \le 2k-i-j$ is equivalent to $i-j\le k-1$ that holds 
    due to the assumption $0<i,j<k$.   Hence, both representations $(i-j)$ and $(i-j+2)$ are present in expansion (\ref{m}) and in expansion (\ref{ij}) that implies 
    \be
     \operatorname{dim}\operatorname{Hom}\left((a)\times(a)\times(i)\times(j),(0)\right) >1 \, .
    \ee

\end{proof}

To summarise, we found that special triples of representations satisfying \eqref{propertysu2} can occur only when the simple current appears in the fusion of the boundary condition label $a$ with itself. From the fusion rules \eqref{su2fus} we notice that this is only possible if $k$ is even $a=k/2$ in which case  $(a)$ is invariant under fusion with the simple current $(k)$. The remaining labels $i,j$ are such that at least one of them is a simple current while the other is any weight in $\mathcal{R}^{2a}_{a,a}$. 
In the next section we prove a similar result for $\mathrm{SU(3)}$ employing essentially the same strategy.\\

    We can use this $\mathrm{SU}(2)$ result to obtain a classification of special triples for minimal models. Consider the $\mathcal{M}_{p,p^\prime}$ minimal model  with Kac table
\begin{equation}
        \mathcal{I}_{p,p^\prime}:=\{(r,s)\,|\, 1\leq r\leq p^\prime-1,\; 1\leq s\leq p-1\} 
    \end{equation}
    where $(r,s)$ denotes a highest-weight Virasoro irreducible representation. Since in this case all representations are self-conjugate and the multiplicities are at most 1, the classification problem reduces to finding triples $(r,s),(r_1,s_1),(r_2,s_2)\in \mathcal{I}_{p,p^\prime}$ of representations such that
    \begin{equation}\label{mini1}
       (r,s)\otimes(r,s) \supset (r_1,s_1) \quad \text{and}\quad (r,s)\otimes (r,s)\supset (r_2,s_2)
    \end{equation}
    and
    \begin{equation}\label{mini2}
        \operatorname{dim}\operatorname{Hom}\left((r_1,s_1)\otimes(r_2,s_2)\otimes(r,s)\otimes(r,s),(1,1)\right)=1\,.
    \end{equation}
     The fusion rule of two irreducible representations factorises into two $\mathrm{SU}(2)$ type fusion rules:
     \begin{equation}\label{minimalfus}
        (r_1,s_1)\otimes(r_2,s_2)=\sum_{r_3=|r_1-r_2|+1}^{\operatorname{min}(r_1+r_2,2p^\prime-r_1-r_2)-1}\,\sum_{s_3=|s_1-s_2|+1}^{\operatorname{min}(s_1+s_2,2p-s_1-s_2)-1}(r_3,s_3)
    \end{equation}
    where the sums increase in steps of 2. Hence, for our classification problem \eqref{mini1}, \eqref{mini2} we can consider the two $\mathrm{SU}(2)$-type triples $(r,r_1,r_2)$ and $(s,s_1,s_2)$ separately. Then minimal model triples $(r,s),(r_1,s_1),(r_2,s_2)$ that satisfy \eqref{mini1} and \eqref{mini2} are such that both the $\mathrm{SU}(2)$ triples $(r,r_1,r_2)$ and $(s,s_1,s_2)$ satisfy the corrersponding $\mathrm{SU}(2)$ property for special triples described in Theorem \ref{su2thm}. Considering all the possible combinations of labels, this leads to new types of special triples not present in the $\mathrm{SU}(2)$ case.  More specifically, we find two types of solutions for triples   $(r,s),(r_1,s_1),(r_2,s_2)$ that satisfy \eqref{mini1} and \eqref{mini2}. The first is of the same form as the $\mathrm{SU}(2)$ result, namely at least one of $(r_1,s_1),(r_2,s_2)$ must be a simple current. If the simple current is $(1,p-1)$ then we have the restriction on the boundary condition label $(r,s)=(r,p/2)$ and $p$ must be even. For the other simple current $(p^\prime -1,1)$ the boundary label must be of the form $(p^\prime/2,s)$ with $p^\prime$ even. For the second type of solution, consider the following subrings of the fusion ring
     \begin{equation}
        R_1:=\{(r,1)\,|\, 1\leq r\leq p^\prime-1\} \quad \text{and}\quad R_2:=\{(1,s)\,|\, 1\leq s\leq p-1\}\,.
    \end{equation}
    Then triples of the second type have arbitrary $(r,s)$, and $\left((r_1,s_1),(r_2,s_2)\right)\in R_1\times R_2$. The symmetric case  $\left((r_1,s_1),(r_2,s_2)\right)\in R_2\times R_1$ is also a solution.

\section{Special triples for $\mathrm{SU(3)}$ WZW theory}\label{sec4}

In this section we consider the special triples  problem  for the $\mathrm{SU(3)}$ WZW models. 
We use the explicit Begin–-Mathieu--Walton formula \cite{BMW} (see also \cite{Barker}) to which we refer as the BMW formula. 
We first use it to describe explicitly the fusion of a representation with its conjugate where we separate all multiplicity 1 representations. 
Geometrically those reside on the sides of a polygon on a plane while the higher multiplicity representations fill in the bulk of the polygon. 
We then adapt a strategy similar to the one we employed for $\mathrm{SU(2)}$. Namely, we show that for any two representations $\hat \mu$, $\hat \nu$
that appear in the fusion $\hat \lambda \times \overline{\hat \lambda}$ with multiplicity 1 and that are not simple currents, the fusions $\hat \lambda \times \hat \mu$ and $\overline{ \hat \lambda} \times \hat \nu$ always have at least two common irreducible representations.   This shows that all special triples come from simple currents. The analysis is more involved as, unlike  the 
$\mathrm{SU(2)}$ case where the geometry is one-dimensional, we have to deal with two-dimensional polygons and the common representations depend on the level $k$ and on the boundary segments.

\subsection{The BMW formula }

Let $P^k_+$ be the set of $\aff$ dominant integral weights at level $k$ and let $\hat{\lambda},\hat{\mu},\hat{\nu}\in P^k_+$ be the highest weights  of three integrable $\widehat{\mathfrak{su}}(3)_k$ representations $V_{\hat{\lambda}},V_{\hat{\mu}},V_{\hat{\nu}}$. The corresponding  $\mathfrak{su}(3)$ highest weights are $\lambda=(\lambda_1,\lambda_2),\,\mu=(\mu_1,\mu_2),\,\nu=(\nu_1,\nu_2)$. The  $\aff$ multiplicity of the representation $\hat{\nu}$ in the fusion product $\hat{\lambda} \times \hat{\mu}$ according to the BMW formula  \cite{BMW}
is
\begin{equation}
   N_{\lambda, \mu}^{(k) \nu}= \begin{cases}\min \left(k_0^{\max }, k\right)-k_0^{\min }+1 & \text { if } k \geq k_0^{\min } \text { and } N_{\lambda, \mu}^\nu>0, \\ 0 & \text { if } k<k_0^{\min } \text { or } N_{\lambda, \mu}^\nu=0,\end{cases}
\end{equation}
where
\begin{align} \label{BMWkk}
& k_0^{\min }=\max \left(\lambda_1+\lambda_2, \mu_1+\mu_2, \nu_1+\nu_2,\mathcal{A}-\lambda_1,\mathcal{A}-\mu_1,A-\nu_2,\mathcal{B}-\lambda_2,\mathcal{B}-\mu_2,\mathcal{B}-\nu_1\right)\, ,  \nonumber \\
& k_0^{\max }=\min (\mathcal{A}, \mathcal{B}) \, , \nonumber \\
& \mathcal{A}=\frac{1}{3}\left[2\left(\lambda_1+\mu_1+\nu_2\right)+\lambda_2+\mu_2+\nu_1\right]\, , \nonumber \\
& \mathcal{B}=\frac{1}{3}\left[\lambda_1+\mu_1+\nu_2+2\left(\lambda_2+\mu_2+\nu_1\right)\right]
\end{align}
and $N_{\lambda,\mu}^{\nu}$ are the $\mathfrak{su}(3)$ tensor product multiplicities given by
\begin{equation}
    N_{\lambda,\mu}^{\nu}=\left(k_0^{\text{max}}-k_0^{\text{min}}+1\right)\delta
\end{equation}
with
\begin{equation}
    \delta= \begin{cases}1 & \text { if } \quad k_0^{\text {max }} \geq k_0^{\text {min }} \quad \text { and } \quad \mathcal{A}, \mathcal{B} \in \mathbb{Z}_{+} \\
0 & \text { otherwise }\end{cases}
\end{equation}

\subsection{Direct sum decomposition of  \texorpdfstring{$(a,b)\times(b,a)$}{(a,b)x(b,a)} } \label{sec4.2}

We are interested in the $\aff$ fusion product of a representation with its conjugate. We shall use the BMW formula with $\lambda=(a,b)$ and $\mu=\overline{\lambda}=(b,a)$. By the symmetry of the fusion product we  restrict ourselves without loss of generality to the case $a\geq b$. Using the BMW formula  we obtain
\begin{equation} \label{abba}
    (a,b)\times (b,a) =\bigoplus_{i=1}^5\bigoplus_{R\in T_i} N^{(k)}_{(a,b),(b,a)}(p,l) R 
\end{equation}
where the multiplicities $N^{(k)}_{(a,b),(b,a)}$ are given by formula  \eqref{offdiagmult} and the representations $R$  are taken from five sets   $T_i$ 
specified in \eqref{offdiag2}, \eqref{offdiag3}, \eqref{offdiag4}, \eqref{diag2} and \eqref{diag3}.

We now explain the details of formula (\ref{abba}). The representations $\nu$ that appear in $(a,b)\times (b,a)$ can be split into those who have equal Dynkin labels which we call diagonal and the remaining ones which we call off-diagonal. %\\
%\noindent \underline{Off-diagonal}\\
We first consider the off-diagonal representations $(\nu_1,\nu_2)$, $\nu_1\ne \nu_2$.
% which appear in the $\mathfrak{su}(3)$ tensor product $(a,b)\otimes (b,a)$.   
According to the BMW formula, they should  satisfy $k_0^{\text{max}}\geq k_0^{\text{min}}$ which singles out the representations which appear in the $\mathfrak{su}(3)$ tensor product. Solving this condition  we find that $(\nu_1,\nu_2)$ must be of the form
\begin{equation}\label{offdiag1}
    (p-2l,p+l) \quad \text{or} \quad (p+l,p-2l) \quad \text{with}\quad  2\leq p\leq a+b,\;1\leq l\leq \operatorname{min}\left(\floor*{p/2},b\right)
\end{equation}
where also $p,l\in \mathbb{N}$. For these representations we have
\begin{equation}
    k_0^{\text{max}}=a+b+p-l
\end{equation}
and
\begin{equation}
    k_0^{\text{min}}=a+b-l+\operatorname{max}(p-b,l)+\operatorname{max}(p-a,l)
\end{equation}
To find the representations from \eqref{offdiag1} that appear in the $\aff$ fusion $(a,b)\times (b,a)$ we have to further solve the inequality
\begin{equation}
    k \geq k_0^{\text{min}}\,.
\end{equation}
This gives the following three non-overlapping sets:
\begin{align}\label{offdiag2}
  & \text{if}\quad  b\geq 2 \quad \text{and}\quad  k\geq 2a+3  \nonumber \\
   &T_1= \bigg\{(p-2l,p+l),(p+l,p-2l)\,|\, &a+2\leq p\leq \operatorname{min}\left(a+b,k-a-1,\floor*{\frac{2k}{3}},\floor*{\frac{b+k}{2}}\right), \nonumber \\
    & &\operatorname{max}(1,2p-k)\leq  l\leq \operatorname{min}\left(p-a-1,\floor*{\frac{p}{2}},b\right)\bigg\}\, ,
\end{align}
    \begin{align}\label{offdiag3}
   & \text{if} \quad b\geq 1 \quad \text{and} \quad k\geq a+b+1 \nonumber \\
     &  T_2=\bigg\{(p-2l,p+l),(p+l,p-2l)\,|\, &&b+1\leq p\leq \operatorname{min}\left(a+b,k-a\right), \nonumber \\
        & &&\operatorname{max}(1,p-a)\leq  l\leq \operatorname{min}\left(p-b,\floor*{\frac{p}{2}},b\right)\bigg\}\, ,
    \end{align}
     \begin{align}\label{offdiag4}
    & \text{if} \quad b\geq 2 \quad \text{and} \quad k\geq a+b+1 \nonumber\\
     &  T_3=\bigg\{(p-2l,p+l),(p+l,p-2l)\,|\, && 2\leq  p\leq \operatorname{min}\left(k-a-1,2b-2\right), \nonumber \\
         & && \operatorname{max}(1,p-b+1)\leq  l\leq \operatorname{min}\left(k-a-b,\floor*{\frac{p}{2}},b\right) \bigg\}\, .
    \end{align}
    If the given condition on $k$, $a$ and $b$ is not satisfied the corresponding set is empty.
    The above sets are all empty for  $b=0$ as in this case there are no off-diagonal representations appearing in the fusion $(a,0)\times (0,a)$. 
    
It follows from the BMW formula that the off-diagonal representations in the sets $T_{i}$  appear with multiplicity
\begin{equation}\label{offdiagmult}
    N^{(k)}_{(a,b),(b,a)}(p,l)=\operatorname{min}(k,a+b+p-l)-a-b+l-\operatorname{max}(p-b,l)-\operatorname{max}(p-a,l)+1\, .
\end{equation}
This expression satisfies the identities 
\begin{equation}
N^{(k)(p-2l,p+l)}_{(a,b),(b,a)}=N^{(k)(p+l,p-2l)}_{(a,b),(b,a)}= N^{(k)}_{(a,b),(b,a)}(p,l)
\end{equation}
that imply that  given  a representation appears in $(a,b)\times (b,a)$  its conjugate also appears with the same multiplicity.

%\noindent \underline{Diagonal}\\
We parametrise an arbitrary diagonal representation as $(\nu_1,\nu_2)=(p,p)$ with $p\in \mathbb{N}$. To find which of these appear in the $\aff$ fusion $(a,b)\times (b,a)$ we have to solve
\begin{equation}\label{diag1}
    \operatorname{min}(k_0^{\text{max}},k) \geq k_0^{\text{min}}\,.
\end{equation}
where 
%Using (\ref{BMWkk}) we find
% For the above specified values of $(\lambda_1,\lambda_2),(\mu_1,\mu_2)$ and $(\nu_1,\nu_2)$ we find
\begin{equation}
    k_0^{\text{max}}=a+b+p \qquad \text{and}\qquad k_0^{\text{min}}=a+b+p-\operatorname{min}(p,b,a+b-p)\, .
\end{equation}
 Solving  \eqref{diag1} for the admissible values of $p$ we find two disjoint sets
    \begin{equation}\label{diag2}
       T_4=\{(p,p)\,|\, 0\leq p\leq \operatorname{min}(a,k-a)\},
    \end{equation}
    \begin{equation}\label{diag3}
       \text{for}\quad k\geq 2a+2:\quad T_5= \bigg\{(p,p)\,|\, a+1\leq p\leq \operatorname{min}\left(a+b,\floor*{\frac{k}{2}}\right)\bigg\}\,
    \end{equation}
    and the sets are empty if the condition on $k$ is not satisfied. 
The representations $(p,p)\in T_4\cup T_5$ appear in $(a,b)\times (b,a)$ with multiplicity
\begin{equation}\label{diagmult}
    N_{(a,b),(b,a)}^{(k)}(p,0)=\operatorname{min}(k,a+b+p)-a-b-\operatorname{max}(p-b,0)-\operatorname{max}(p-a,0)+1
\end{equation}
which is obtained from \eqref{offdiagmult}  by setting $l=0$. 
Note that all sets $T_{i}$, $1\le i\le 5$ are empty if $k<a+b$ which is the standard unitarity bound for $\mathrm{SU(3)}$.

\mycomment{
We can now  express the set of all representations appearing in the direct sum decomposition of $(a,b)\times (b,a)$ as
\begin{equation}
    \mathcal{R}^k_{(a,b),(b,a)}= \bigcup_{i=1}^5 T_i\,.
\end{equation} 
}

\subsection{Multiplicity 1 representations in  \texorpdfstring{$(a,b)\times(b,a)$}{(a,b)x(b,a)} }\label{sec4.3}

To find the representations with multiplicity $1$ in the fusion $(a,b)\times (b,a)$ we solve $N^{(k)}_{(a,b),(b,a)}(p,l)=1$ and $N^{(k)}_{(a,b),(b,a)}(p,0)=1$ in each set \eqref{offdiag2}, \eqref{offdiag3},\eqref{offdiag4} and \eqref{diag2}, \eqref{diag3} respectively. 
The solution can be described as the following   union of pairwise disjoint sets 
\begin{equation}\label{fullmult1}
    \mathcal{R}^{(1);k}_{(a,b),(b,a)}= \left(\bigcup_{i=1}^3 \left( Q_i\cup \overline{Q}_i\cup C_i\cup \overline{C}_i\right) \right)\cup \left(\bigcup_{i=4}^6 Q_i\right)\cup \left(\bigcup_{i=4}^6 C_i\right)
\end{equation}
where
\begin{align}\label{d1}
        &\text{for}\quad b\geq 2, \quad 2a+3\leq k\leq 2a+2b-1 \nonumber \\
        &Q_1=\bigg\{(2k-3p,3p-k)\,|\, \ceil*{\frac{k+1}{2}}\leq p\leq \operatorname{min}\left(a+b,k-a-1,\floor*{\frac{2k}{3}},\floor*{\frac{b+k}{2}}\right)\bigg\}\,,
    \end{align}
    \begin{align}\label{d2}
        &\text{for}\quad  b\geq 1, \quad a+b+1\leq k\leq 2a+b \nonumber \\
        &Q_2=\bigg\{(k-a-2l,k-a+l)\,|\, \operatorname{max}(1,k-2a)\leq l\leq \operatorname{min}\left(k-a-b,\floor*{\frac{k-a}{2}},b\right)\bigg\}\,,
    \end{align}
 \begin{align}\label{d3}
        &\text{for}\quad b\geq 2, \quad a+b+1\leq k\leq a+2b-1 \nonumber \\
        &Q_3=\{(p-2k+2a+2b,p+k-a-b)\,|\, 2k-2a-2b\leq p\leq k-a-1\}\,,
    \end{align}
     \begin{align}\label{c1}
        &\text{for}\quad b\geq 2, \quad  k\geq 2a+b+2 \nonumber \\
        &C_1=\{(a+b-2l,a+b+l)\,|\, \operatorname{max}(1,2a+2b-k+1)\leq l \leq b-1\}\,,
    \end{align}
     \begin{align}\label{c2}
        &\text{for}\quad b\geq 1, \quad k\geq a+2b+1 \nonumber \\
        &C_2=\{(p-2b,p+b)\,|\, 2b\leq p\leq \operatorname{min}(a+b,k-a-1)\}\,,
    \end{align}
     \begin{align}\label{c3}
        &\text{for}\quad b\geq 2, \quad  k\geq a+b+2 \nonumber \\
        &C_3=\{(0,3l)\,|\, 1\leq l\leq \operatorname{min}\left(k-a-b-1,b-1\right)\}\,,
    \end{align}
     \begin{align}\label{d4}
        &\text{for}\quad k=a+b \nonumber \\
        &Q_4=\{(p,p)\,|\, 1\leq p\leq b\}\,,
    \end{align}
     \begin{align}\label{d5}
        &\text{for}\quad a\neq b,\quad  a+b+1 \leq k\leq 2a \nonumber \\
        &Q_5=\{(k-a,k-a)\}\,,
    \end{align}
     \mycomment{
     \begin{align}\label{d6}
        &\text{for}\quad k=2a \nonumber \\
        &Q_6=\{(a,a)\}\,,
    \end{align}
    }
    \begin{align}\label{d7}
        &\text{for}\quad b\geq 1,\quad  2a+2 \leq k\leq 2a+2b \quad \text{and} \quad k \quad \text{even} \nonumber \\
        &Q_6=\bigg\lbrace\left(\frac{k}{2},\frac{k}{2}\right)\bigg \rbrace\,,
    \end{align}
    \begin{align}\label{c4}
        & \text{for}\quad  k\geq a+b\nonumber \\
        &C_4=\{(0,0)\}\,,
    \end{align}
      \begin{align} \label{c5}
        &\text{for}\quad a\geq 2,\, b=0,\quad  k\geq a+1 \nonumber \\
        &C_5=\{(p,p)\,|\, 1\leq p\leq \operatorname{min}(a-1,k-a)\}\,,
    \end{align}
     \begin{align}\label{c6}
        & \text{for}\quad  k\geq 2a+2b+1 \nonumber \\
        &C_6=\{(a+b,a+b)\}
    \end{align}
Each of  the above sets is set to be  empty if the stipulated conditions involving $a,b,k$ are not satisfied. For each off-diagonal set $Q_i,C_i$ with $i\in\{1,2,3\}$ its conjugate $\overline{Q}_i,\overline{C}_i$ also appears with multiplicity $1$. The sets were intuitively named $C_i$ from ``classical'' and $Q_i$ from ``quantum'' as the former contain weights that satisfy $k\geq \kmax$ while the latter have weights with $ k<\kmax$. For large enough $k$, namely for $k\geq 2a+2b$, all the multiplicity $1$ representations in $(a,b)\times (b,a)$ will be exactly the multiplicity $1$ representations in the $\mathfrak{su}(3)$ tensor product $(a,b)\otimes (b,a)$. Geometrically, we can interpret each element in the sets $Q_i,\overline{Q}_i,C_i,\overline{C}_i$ for fixed level as a point in an $\mathfrak{su}(3)$ highest weight lattice. Then the union of all these sets of multiplicity $1$ representations forms the boundary of an $n$-sided polygon in that lattice. Depending on the level and the values of $a$ and $b$ 
we get $3\le n \le 7$. 
We showcase this in Figures \ref{graph1}, \ref{graph2} and \ref{graph3} where the red coloured dots are representations which appear in the direct sum decomposition of $(35,25)\times (25,35)$ with multiplicity $1$, while the blue dots have multiplicity greater than $1$. On those figures we also mark the multiplicity 1 regions $Q_i$ and $C_i$ and their conjugates.  

\begin{figure}[H]
\centering
\includegraphics[width=12cm]{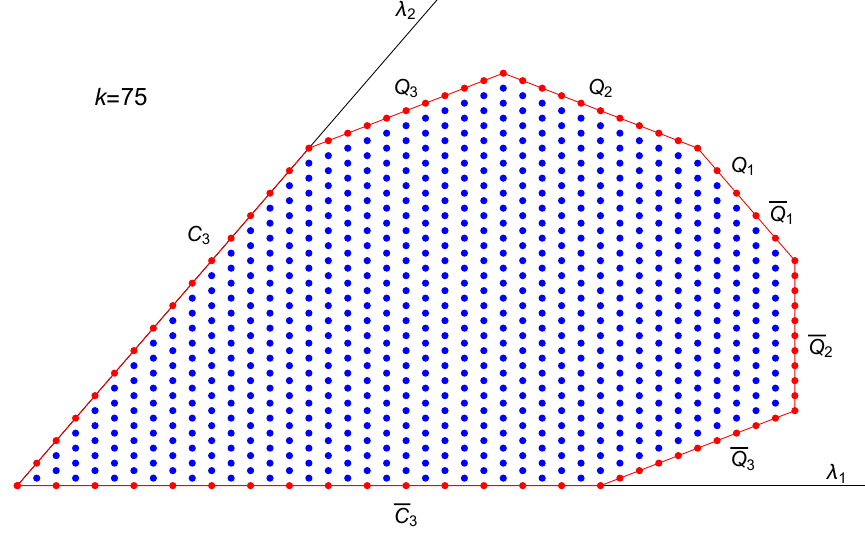}
\caption{Highest weights $(\lambda_1,\lambda_2)$ of representations that appear in the $\widehat{\mathfrak{su}}(3)_k$ fusion $(35,25)\times(25,35)$ at $k=75$.}\label{graph1}
\end{figure}

\begin{figure}[H]
\centering
\includegraphics[width=12cm]{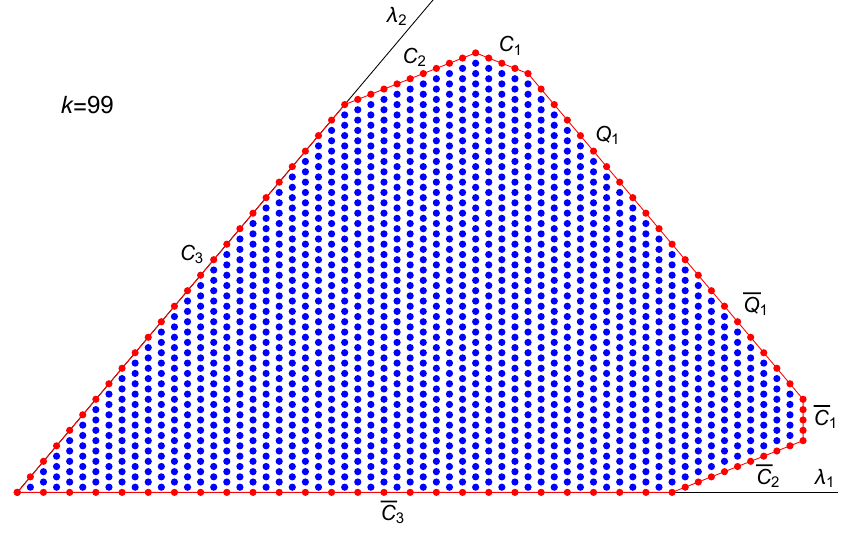}
\caption{Highest weights $(\lambda_1,\lambda_2)$ of representations that appear in the $\widehat{\mathfrak{su}}(3)_k$ fusion $(35,25)\times(25,35)$ at $k=99$.}\label{graph2}
\end{figure}

\begin{figure}[H]
\centering
\includegraphics[width=12cm]{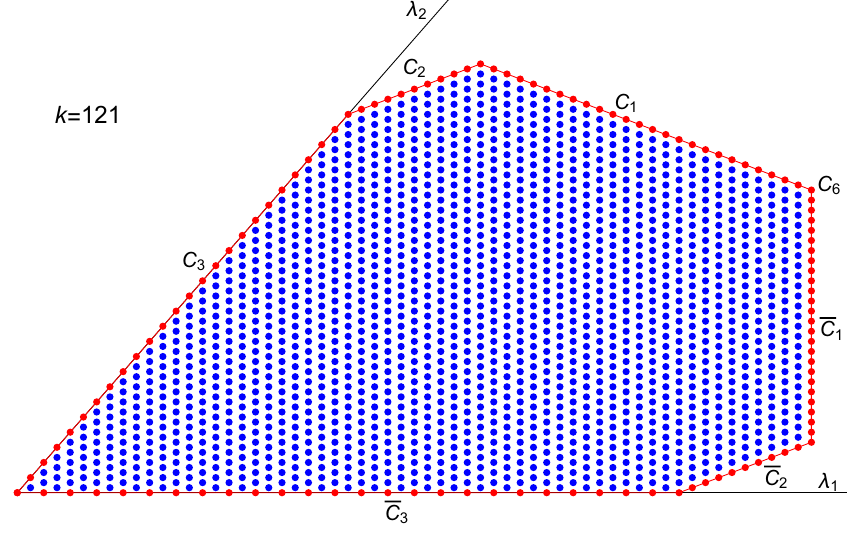}
\caption{Highest weights $(\lambda_1,\lambda_2)$ of representations that appear in the $\widehat{\mathfrak{su}}(3)_k$ fusion $(35,25)\times(25,35)$ at $k=121$.}\label{graph3}
\end{figure}

%\subsection{Simple currents in \texorpdfstring{$(a,b)\times(b,a)$}{(a,b)x(b,a)} fusion }

\subsection{Special triples exist only when a simple current is present} \label{sec4.4}

While the trivial representation that corresponds to the identity simple current always appears in the fusion $(a,b)\times (b,a)$, the non-trivial 
simple currents $\{(k,0),(0,k)\}$ only appear in a special case described in the following lemma which is straightforward to establish using the BMW formula and 
the description the sets $T_{i}$. 
% Recall that all the off-diagonal representations that appear in the fusion $(a,b)\times (b,a)$ are described by the  sets $T_1,T_2,T_3$ given in \eqref{offdiag2}, \eqref{offdiag3} and \eqref{offdiag4} respectively. The following result is useful in order to identify when the simple currents can appear in $\mathcal{R}^k_{(a,b),(b,a)}$.
\begin{lemma}\label{lemma2}
    Let $T_1,T_2,T_3 \subseteq \mathcal{R}_{(a,b),(b,a)}^k$ be the off-diagonal representations that appear in the direct sum decomposition of the  $\aff$ fusion $(a,b)\times (b,a)$. Then 
   \begin{enumerate}[i),wide, labelwidth=!, labelindent=0pt]\item \begin{equation}
        \{(k,0),(0,k)\}  \nsubseteq T_1\cup T_3
    \end{equation}
    and \item
    \begin{equation}
          \{(k,0),(0,k)\}  \subseteq T_2 \Longleftrightarrow a=b \quad \text{and}\quad k=3a \, .
    \end{equation}
    \end{enumerate}
\end{lemma}
Notice that the simple currents are in the sets of multiplicity $1$ representations $Q_2$ and $\overline{Q}_2$ given by \eqref{d2}.

\mycomment{
It will be convenient to isolate the simple currents. To achieve this we split $D_2$ according to the level $k$ as
 \begin{align}\label{d2low}
        &b\geq 1, \quad a+b+1\leq k\leq a+2b-1 \nonumber \\
        &D_2^{\text{low}}=\bigg\{(k-a-2l,k-a+l)\,|\, \operatorname{max}(1,k-2a)\leq l\leq k-a-b\bigg\}
    \end{align}
    and
     \begin{align}\label{d2high}
        &b\geq 1, \quad a+2b\leq k\leq 2a+b \nonumber \\
        &D_2^{\text{high}}=\bigg\{(k-a-2l,k-a+l)\,|\, \operatorname{max}(1,k-2a)\leq l\leq b\bigg\}\,.
    \end{align}
In this way we see that whenever the simple currents $(0,k)$ and $(k,0)$ appear in the fusion $(a,b)\times (b,a)$ they will be the only elements of the sets $D_2^{\text{high}}$ and $\overline{D}_2^{\text{high}}$ respectively.
}

For future reference let us note that the fusion of a simple current with an arbitrary reprsentation is given by
\begin{equation}\label{currentaction}
    (k,0)\times (i,j) =(k-i-j,i) \qquad \text{and} \qquad (0,k)\times (i,j) = (j,k-i-j)\,.
\end{equation}
Finally we introduce the following notation for the set of $\widehat{\mathfrak{su}}(3)_k$ simple currents: 
\begin{equation}
    \mathcal{J}_3=\{J_0,J_1,J_2\} \qquad\text{where} \quad J_0=(0,0),\; J_1=(k,0),\; J_2=(0,k)\,.
\end{equation}

\mycomment{
as well as the fusion action of the simple currents on a set containing $\aff$ representations $\mathcal{S}=\{s_1,s_2,\ldots,s_l\}$:
\begin{equation}
    \mathcal{J}\cdot \mathcal{S}:= \{J_m\times s_n\,|\, m=0,1,2,\; n=1,2,\ldots,l\}
\end{equation}
}

We are now ready to present the main result of this section that establishes that all special triples in $\mathrm{SU(3)}$ WZW theory come from simple currents. 
\begin{thm}\label{thmsu3}
    Let $(a,b)$ be an integrable $\aff$ representation that is not a simple current\footnote{If $(a,b)$ is any simple current then we have a trivial situation since $\mathcal{R}^{(1);k}_{(a,b),(b,a)}=\{(0,0)\}$}  and let $(i,j),(m,n)\in \mathcal{R}^{(1);k}_{(a,b),(b,a)}$. Then if $(i,j),(m,n)\notin \mathcal{J}_3$ the following holds
    \begin{equation}
        \operatorname{dim}\operatorname{Hom}\left((a,b)\times (b,a)\times (i,j)\times (m,n),(0,0)\right) > 1
    \end{equation}
\end{thm}
\begin{proof}
\mycomment{    We know from the discussion in the previous section that $(i,j),(m,n)\notin \{(k,0),(0,k)\}$ is always satisfied for $a\neq b$ since then there are no non-trivial simple currents in $\mathcal{R}^{(1);k}_{(a,b),(b,a)}$. Hence, it is  only relevant if $a=b$ and in any case can be rewritten as
    \begin{equation}\label{ijsets}
        (i,j),(m,n)\in \mathcal{R}^{(1);k}_{(a,b),(b,a)}\setminus\mathcal{J}
    \end{equation}
}
\mycomment{
\begin{equation}\label{ijsets}
     (i,j),(m,n)\in
    \left\{\begin{array}{lr}
       \mathcal{R}^{(1);k}_{(a,a),(a,a)}\setminus\left(D_2^{\text{high}}\cup \overline{D}_2^{\text{high}}\cup\{(0,0)\}\right), & \text{if } a=b\\
        \mathcal{R}^{(1);k}_{(a,b),(b,a)}\setminus\{(0,0)\}, & \text{if } a\neq b
        \end{array}\right. \,.
\end{equation}
}
By the symmetry of the fusion product we restrict without loss of generality to the case $a\geq b$. Since $(i,j),(m,n)\in \mathcal{R}^{(1);k}_{(a,b),(b,a)}\setminus\mathcal{J}_3$ we have
\begin{align}
    (a,b)\times (i,j) &= (a,b) \oplus\cdots \qquad \text{and} \nonumber \\
    (b,a)\times (m,n) &=(b,a)\oplus \cdots
\end{align}
that implies 
%\begin{align}
%    (a,b)\times (b,a)\times (i,j)\times (m,n) &= [(a,b)\times (i,j)]\times [(b,a)\times (m,n)] \nonumber \\
 %   & =[(a,b)\oplus \cdots ]\times [(b,a)\oplus \cdots]
%\end{align}
%Since $(a,b)\times (b,a)=(0,0)\oplus\cdots$ we conclude  that 
\begin{equation}
     \operatorname{dim}\operatorname{Hom}\left((a,b)\times (b,a)\times (i,j)\times (m,n),(0,0)\right) \geq 1 \, .
\end{equation}
Since the set $\mathcal{R}^{(1);k}_{(a,b),(b,a)}$ is symmetric under conjugation, to prove the theorem it suffices to prove 
that for each $(i,j)\in \mathcal{R}^{(1);k}_{(a,b),(b,a)}\setminus\mathcal{J}_3$ and each $(m,n)\in \mathcal{R}^{(1);k}_{(a,b),(b,a)}\setminus\mathcal{J}_3$
there exists at least one representation $(c,d) \ne (a,b) $
%Then to prove the theorem we have to show  $ \operatorname{dim}\operatorname{Hom}\left((a,b)\times (b,a)\times (i,j)\times (m,n),(0,0)\right) \neq 1$. It suffices to show that for all $(i,j)\in \mathcal{R}^{(1);k}_{(a,b),(b,a)}\setminus\mathcal{J}_3$    there exists at least one representation $(c,d)$
 such that
\begin{equation}\label{claim1}
(c,d)\in \mathcal{R}^k_{(a,b),(i,j)}\cap \mathcal{R}^k_{(a,b),(m,n)} \, .
\end{equation}
\mycomment{or in other words there exists $(c,d)$ such that
\begin{align}\label{claim1}
    (a,b)\times (i,j) &= (a,b) \oplus N^{(k)(c,d)}_{(a,b),(i,j)} (c,d) \oplus\cdots \qquad \text{and} \nonumber \\
    (a,b)\times (n,m) &=(a,b)\oplus N^{(k)(c,d)}_{(a,b),(n,m)} (c,d)\oplus \cdots
\end{align}
with non-zero $ N^{(k)(c,d)}_{(a,b),(i,j)}$ and $N^{(k)(c,d)}_{(a,b),(n,m)}$.} 
%In order to prove this, we need to provide a representation $(c,d)$ for each possible combination of choices of $(i,j)$ and $(m,n)$. 
Using \eqref{fullmult1}  we see that 
\begin{equation}
    (i,j),(m,n)\in \left(\bigcup_{i=1}^3 \left( D_i\cup \overline{D_i}\cup C_i\cup \overline{C_i}\right) \right)\cup \left(\bigcup_{i=4}^6 D_i\right)\cup \left(\bigcup_{i=4}^6 C_i\right) \setminus\mathcal{J}_3\,.
\end{equation}
The representations  $(c,d)$ will depend on which set $(i,j)$ and $(m,n)$ belong to. 
We find that the following six possible values for $(c,d)$ are sufficient: %are the sufficient representations of the type $(c,d)$:
\begin{align}
    & w_1=(a-1,b+2), \qquad w_2=(a+2,b-1), \qquad w_3=(a+1,b-2) \nonumber \\
    & w_4= (a-2,b+1), \qquad 
     w_5=(a+1,b+1), \qquad w_6=(a-1,b-1)\,.
\end{align}
In table \ref{table1}, table \ref{table2} and table \ref{table3} we give a representation $(c,d)\in\{w_1,\ldots,w_6\}$ that satisfies \eqref{claim1} for all possible pairs of sets that $(i,j),(m,n)$ belong to.

\begin{table}[H]
\begin{center}
\begin{tabular}{ |a|c|c|c|c|c|c|c|c|c|c|c|c| } 
 \hline
 \rowcolor{eslb}
 & $Q_1$  & $Q_2$    & $Q_3$  & $C_1$ & $C_2$ & $C_3$ & $\overline{Q}_1$ 
 &$\overline{Q}_2$  &$\overline{Q}_3$ &$\overline{C}_1$ &$\overline{C}_2$ & $\overline{C}_3$ \\
 \hline
 $Q_1$ & $w_1$  & $\color{red} w_1,w_6$  & $w_1$ & $w_1$ & $w_5$ & $w_1$ & $w_1$ & $\color{red} w_4,w_6$ & $w_4$ & $w_1$ & $w_5$ & $w_1$ \\
 \hline
 $Q_2$ & $\color{red} w_1,w_4,w_6$ & $\color{red}w_1,w_4,w_6$  & $\color{red}w_6$ & - & $\color{red} w_1$ & $\color{red}w_4,w_6$ & $\color{red}w_4,w_6$ & $\color{red}w_4,w_6$  & $\color{red}w_6$ & - & $\color{red}w_1$ & $\color{red}w_4,w_6$ \\
 \hline
 $Q_3$ & $w_3$ & $\color{red}w_6$  & $w_6$ & - & - & $w_6$ & $w_4$ & $\color{red}w_6$  & $w_6$ & - & - & $w_6$ \\
 \hline
 $C_1$ & $w_1$ & -  & - & $w_1$ & $w_5$ & $w_1$ & $w_1$ & -  & - & $w_1$ & $w_5$ & $w_1$\\
 \hline
 $C_2$ & $w_1$ &  $\color{red}w_1$ & - & $w_5$ & $w_5$ & $w_1$ & $w_1$ &  $\color{red}w_1$ & - & $w_1$ & $w_5$ & $w_1$\\
 \hline
 $C_3$ & $w_1$ & $\color{red}w_1,w_6$  & $w_6$ & $w_1$ & $w_5$ & $w_1$ & $w_1$ & $\color{red}w_4,w_6$  & $w_6$ & $w_1$ & $w_5$ & $w_1$ \\
 \hline
 $\overline{Q}_1$ & $w_1$ & $\color{red}w_1,w_6$  & $w_1$ & $w_1$ & $w_5$ & $w_1$ & $w_1$ & $ \color{red} w_1,w_4,w_6$  & $w_3$ & $w_1$ & $w_5$& $w_1$ \\
 \hline
 $\overline{Q}_2$ & $\color{red}w_1,w_6$ & $\color{red}w_1,w_6$  & $\color{red}w_6$ & - & $\color{red}w_1$& $\color{red}w_1,w_6$ & $\color{red} w_1,w_6$ & $\color{red}w_1,w_4,w_6$  &  $\color{red}w_6$ & - & $\color{red}w_1$ & $\color{red}w_1,w_6$   \\
 \hline
 $\overline{Q}_3$ & $w_1$ & $\color{red}w_6$ & $w_6$ & - & - & $w_1$ & $w_1$ & $\color{red}w_6$   & $w_6$ & - & - & $w_1$   \\
 \hline
$\overline{C}_1$ & $w_1$  & - & - & $w_1$ & $w_5$ & $w_1$ & $w_1$ & -  & - & $w_1$ & $w_5$ & $w_1$  \\
\hline
$\overline{C}_2$ & $w_5$  & $\color{red}w_1$ & - & $w_5$ & $w_5$ & $w_5$ & $w_5$  & $\color{red}w_1$ & -  & $w_5$ & $w_5$ & $w_5$  \\
\hline
$\overline{C}_3$ & $w_1$ & $\color{red}w_1,w_6$  & $w_6$ & $w_1$ & $w_5$ & $w_1$ & $w_1$ & $\color{red}w_4,w_6$  & $w_6$ & $w_1$ & $w_5$ & $w_1$ \\
\hline
\end{tabular}
\end{center}
\caption{Overlap between the off-diagonal sets. Each entry gives the representation $(c,d)$ that satisfies \eqref{claim1} when $(i,j)$ belongs to the set specified by the leftmost column and $(m,n)$ belongs to the set specified by the topmost row.}
\label{table1}
\end{table}

\begin{table}[H]
\begin{center}
\begin{tabular}{ |a|c|c|c|c|c|c|c|c|c|c|c|c|c|c| } 
 \hline
 \rowcolor{eslb}
& $Q_4$ & $Q_5$  & $Q_6$ & $C_5$ & $C_6$\\
\hline
$Q_4$ & $w_6$ & -  & - & - & -   \\
\hline
$Q_5$ & - & triv. &  - & $w_1$ & - \\ 
\hline
$Q_6$ &- & - &  triv. & - & -  \\ 
\hline
$C_5$ & - & $w_1$  & - & $w_1$ & $w_1$ \\
\hline 
$C_6$ & - & - &  - & $w_1$ & triv. \\
\hline
\end{tabular}
\end{center}
\caption{Overlap between the diagonal sets.  Each entry gives the representation $(c,d)$ that satisfies \eqref{claim1} when $(i,i)$ belongs to the set specified by the leftmost column and $(m,m)$ belongs to the set specified by the topmost row.}
\label{table2}
\end{table}

\begin{table}[H]
\begin{center}
\begin{tabular}{ |a|c|c|c|c|c|c|c|c|c|c|c|c|c|c| } 
 \hline
 \rowcolor{eslb}
& $Q_1$  & $Q_2$    & $Q_3$  & $C_1$ & $C_2$ & $C_3$ & $\overline{Q}_1$ 
 &$\overline{Q}_2$  &$\overline{Q}_3$ &$\overline{C}_1$ &$\overline{C}_2$ & $\overline{C}_3$ \\
\hline
$Q_4$   & - & - & - & - & - & - & - & - & - & - & - & - \\
\hline
$Q_5$ & - & $ \color{red}w_1,w_6$  & $w_4$ & - & $\color{red}w_1$ & $w_1$ & - & $\color{red}w_4,w_6$  & $w_4$ & - & $w_1$ & $w_1$ \\
\hline
$Q_6$ & $w_1$ & $\color{red}w_1,w_6$  & $w_1$ & $w_1$ & $w_5$ & $w_1$ & $w_1$ & $\color{red}w_1,w_6$  & $w_2$ & $w_1$ & $w_5$ & $w_1$  \\
\hline
$C_5$ & -  & - & - & - & - & - & - & - & - & - & - & -   \\ 
\hline
$C_6$  & - & - & - & $w_1$ & $w_5$ & $w_1$ & - &  - & - & $w_1$ & $w_5$ & $w_1$ \\
\hline
\end{tabular}
\end{center}
\caption{Overlap between the diagonal and the off-diagonal sets.  Each entry gives the representation $(c,d)$ that satisfies \eqref{claim1} when $(i,i)$ belongs to the set specified by the leftmost column and $(m,n)$ belongs to the set specified by the topmost row.}
\label{table3}
\end{table}
In Tables \ref{table1}, \ref{table2} and \ref{table3} the entries "-" indicate that the two sets do not coexist for fixed values of $k,a,b$. The red coloured entries indicate that the representation $w_l$ appears under additional restrictions involving $k,a,b$ that are imposed such that the two sets in question coexist. The black coloured entries mean that the underlying statement holds generically. The entries "triv." in the table with the diagonal sets are there when $(i,j),(m,n)$ both belong to the same set and that set contains  a single representation. In that case we have $(i,i)=(m,m)$ and the theorem's statement trivially  holds. Lastly, wherever 
a cell in the table contains several choices for $(c,d)$  this means that the overlaps occur for different $w_i$ depending on the level.  We present the minimal set of statements that are needed to obtain the tables and a sample  proof of some of them  in appendix \ref{appendix1}.

\end{proof}

\subsection{Explicit description of all $\mathrm{SU(3)}$ special triples}
Theorem \ref{thmsu3} demonstrates that all special triples in $\mathrm{SU(3)}$ WZW theory must come from simple currents. 
We are going to describe all of them explicitly. Firstly, if $(a,b)\in {\cal J}_{3}$ we have $(i,j)=(m,n)=(0,0)$.
Assuming now $(a,b)\notin {\cal J}_{3}$ the non-identity simple currents must appear in $(a,b)\times (b,a)$. By lemma \ref{lemma2} we must have
 $a=b$ and $k=3a$. The representation $(\frac{k}{3},\frac{k}{3})$ is  fixed under the action of the simple currents and so is the set of representations that appears in its fusion with itself. 
 Geometrically the weights of representations in $(\frac{k}{3},\frac{k}{3})\times (\frac{k}{3},\frac{k}{3})$ fill the lattice points inside an equilateral triangle and the simple current 
 group ${\mathbb Z}_{3}$ acts by rotating this triangle. The representations of multiplicity one belong to the three sides of the triangle and are mapped to themselves. 
 To construct a special triple we take for $(i,j)$ a simple current and for $(m,n)$ any of the multiplicity one representations.

Using the descriptions in subsection \ref{sec4.2}  we obtain
\mycomment{
\begin{equation}\label{class1}
   (m,n)\in  \mathcal{R}^{(1);3a}_{(a,a),(a,a)} = \mathcal{W}_1\cup \overline{\mathcal{W}}_1\cup \mathcal{W}_2\cup \overline{\mathcal{W}}_2
\end{equation}
where
\begin{align}\label{class11}
    \mathcal{W}_1 &= \{(3r,0)\,|\, 0\leq r\leq a-1\} \nonumber \\
   \mathcal{W}_2 & = \bigg\{(6a-3r,3r-3a)\,|\, a\leq r\leq \floor*{\frac{3a}{2}}\bigg\}\,.
\end{align}
By noticing
\begin{equation}
    \mathcal{W}_2\cup \overline{\mathcal{W}}_2 = \{(6a-3r,3r-3a)\,|\, a\leq r\leq 2a\}
\end{equation}
we can rewrite \eqref{class1} equivalently as }
\begin{equation}\label{class2}
     \mathcal{R}^{(1);3a}_{(a,a),(a,a)} = \mathcal{W}_1\cup \overline{\mathcal{W}}_1\cup \mathcal{W}
\end{equation}
where
\begin{align}\label{class3}
 \mathcal{W}_1 & = \{(3r,0)\,|\, 0\leq r\leq a-1\} \quad \text{and}\nonumber \\
    \mathcal{W} &= \{(6a-3r,3r-3a)\,|\, a\leq r\leq 2a\} \, .
\end{align}
Alternatively the above $(i,j)$ and $(m,n)$ representations can be swapped. This describes all possible special triples.

\section{Final comments} \label{sec_comments}
In this paper we looked into the problem of constructing local fields invariant under a topological defect. One way of constructing such objects in theories 
with charge conjugation modular invariant is 
by realising special restrictions on the fusion rule. As explained in the introduction, in the bulk case it leads to the problem of finding the cases of fusion 
of irreducible representations that look like $a\times b = c$ and that was studied in \cite{Buican}. In the boundary case this leads to a different kind 
of condition in three irreducible representations  which we call special triples. In this paper we classified all special triples for $\mathrm{SU(2)}$ and 
$\mathrm{SU(3)}$ WZW theories. We showed that all such triples are related to simple currents. It is our working hypothesis that this holds for all 
$\mathrm{SU(N)}$ WΖW theories. Regarding the proofs it should be noted that while the main theorem for the $\mathrm{SU(2)}$ case was fairly elementary the $\mathrm{SU(3)}$ case was more involved. 
In our opinion it is fair to say that  the $\mathrm{SU(3)}$ problem was tackled by brute force and  it does not look like a good idea trying the same approach for other groups. 
It is known that the representations appearing in fusions always fill discrete subsets inside polytopes. To show that all special triples come from simple currents 
one would need to show that certain polytopes always have a non-empty intersection. Hopefully a geometric proof can be found that does not utilise any explicit formula for the fusion rule. We leave this question to future work. 

Another interesting direction may be to try to construct   subspaces of local boundary or bulk operators invariant under some open topological defects,  whose dimension is greater than 1. 
For example in \cite{paper1} in a rational free boson CFT a four-dimensional real subspace of boundary operators was constructed invariant under 
a set of topological defects.
In the approach  advocated in  the present paper, to find more examples  one could look for triples of representations $(a,i,x)$ for which 
\be \label{higher_dim}
{\rm dim}\, {\rm Hom}(a\times \bar a \times i \times  x\, , \mathbb{1}) = {\rm dim}\, {\rm Hom}(a\times \bar a \times i \, , \mathbb{1}) {\rm dim}\, {\rm Hom}(a\times \bar a  \times  x\, , \mathbb{1}) \, .
\ee
If this identity is realised the subspace of boundary fields corresponding to $i$, whose dimension is  ${\rm dim}\, {\rm Hom}(a\times \bar a \times i \, , \mathbb{1})$ would be a good candidate for an invariant subspace   under passing the open defect $X_{x}$ whose space of junctions with the boundary is of dimension ${\rm dim}\,{\rm Hom}(a\times \bar a  \times  x\, , \mathbb{1})$. We investigated numerically equation (\ref{higher_dim}) for $\mathrm{SU(3)}$ WZW theory and have not found any solutions besides the multiplicity 1 case investigated in section \ref{sec4}. 

Besides relying on the fusion rules the complete solution to the problem of constructing invariant subspaces of local operators should take into account the 
corresponding fusing matrices. One approach to that may be by using the tube algebra structures, see e.g. \cite{Bartsch_Bullimore}, \cite{paper1} for recent discussions.

\begin{center}

{\bf Acknowledgements}
\end{center}

We are indebted to Ingo Runkel for stimulating discussions and to Matt Buican for reading and commenting on the draft.
\appendix

\section{Passing a topological defect through a bulk field in RCFT}\label{appendix0}

In this appendix we use the  TFT approach \cite{FRS1} in order to interpret Figure \ref{figintro1} as an equation between morphisms and obtain the coefficients $Q^{\alpha e \beta}_{i,d}$

Let $\left(\mathcal{C},\times,\mathbb{1}\right)$  denote the modular tensor category encoding the chiral data of the RCFT and let $\{U_i\}_{i\in\mathcal{I}}$ be the simple objects of $\mathcal{C}$ with the convention $U_0\equiv \mathbb{1}$.  To relief notation we will usually write only the label $i$ in place of $U_i$. For these simple objects we fix  bases $\{\lambda^\alpha_{(i,j)k}\}$ in the hom-spaces $\operatorname{Hom}(i\times j,k)$ as well as dual bases $\{\bar{\lambda}_{\alpha}^{(i,j)k}\}$ in $\operatorname{Hom}(k,i\times j)$. These bases will be depicted as
\begin{equation}\label{bases1}
    \vcenter{\hbox{\hspace{2mm}\begin{tikzpicture}[font=\footnotesize,inner sep=2pt]
  \begin{feynman}
  \vertex (j1) at (0,0) [label=below:\(i\)];
  \vertex [above=1.7cm of j1] (j4);
  \draw [fill=yellow] (-0.25,1.8) rectangle (0.75, 1.3);
  \vertex [above=1.3cm of j1] (j3) [label=above:\( \hspace{5mm}\lambda^\alpha_{(i,j)k}\)];
  \vertex [right=0.5cm of j1] (i1) [label=below:\(j\)];
  \vertex (x) at (0.25,3) [label=above:\(k\)];
  \vertex [below=1.2cm of x] (x2);
  \vertex[above=1.7 cm of i1] (i4);
  \vertex[ above=1.3 cm of i1] (i3);
   \diagram*{
    (j1)--[thick] (j3),
    (i1)--[thick] (i3),
    (x2)--[thick] (x)
  };
  \end{feynman}
\end{tikzpicture}}}
~=~
  \vcenter{\hbox{\hspace{2mm}\begin{tikzpicture}[font=\footnotesize,inner sep=2pt]
  \begin{feynman}
\vertex (i1) at (0,0) [label=below:\(i\)];
\vertex [right=1cm of i1] (j1) [label=below:\(j\)];
\vertex [above=0.5cm of i1] (i2);
\vertex [above =0.5cm of j1] (j2);
\vertex  [right=0.5cm of i1] (k1) ;
\vertex [small,orange, dot][above =1.5cm of k1] (k2)[label=right:\(\alpha\)] {};
\vertex [above=1.5cm of k2] (k3) [label=above:\(k\)];
   \draw [thick,rounded corners=1mm] (i1)--(i2)--(k2);
      \draw [thick,rounded corners=1mm] (j1)--(j2)--(k2);
   \diagram*{
(k2)--[thick] (k3)
  };
  \end{feynman}
\end{tikzpicture}}}
\hspace{20mm}
 \vcenter{\hbox{\hspace{2mm}\begin{tikzpicture}[font=\footnotesize,inner sep=2pt]
  \begin{feynman}
  \vertex (j1) at (0,0) [label=above:\(i\)];
  \vertex [below=1.3cm of j1] (j2);
  \draw [fill=yellow] (-0.25,-1.8) rectangle (0.75, -1.3);
  \vertex [below=1.8cm of j1] (j3) [label=\( \hspace{5mm}\bar{\lambda}_{\alpha}^{(i,j)k}\)];
  \vertex [right=0.5cm of j1] (i1) ;
  \vertex[below=0.04cm of i1] (xx) [label=above:\(j\)];
  \vertex (x) at (0.25,-3) [label=below:\(k\)];
  \vertex [above=1.2cm of x] (x2);
  \vertex[below=1.3cm of i1] (i2);
   \diagram*{
    (j1)--[thick] (j2),
    (i1)--[thick] (i2),
    (x)--[thick] (x2)
  };
  \end{feynman}
\end{tikzpicture}}}
~=~
  \vcenter{\hbox{\hspace{2mm}\begin{tikzpicture}[font=\footnotesize,inner sep=2pt]
  \begin{feynman}
\vertex (i1) at (0,0) [label=below:\(k\)];
\vertex [small,orange, dot][above=1.5cm of i1] (i2) [label=right:\(\alpha\)]{};
\vertex [above=3cm of i1] (k1);
\vertex [right =0.5cm of k1] (k2) ;
\vertex[below=0.04cm of k2] (xx) [label=above:\(j\)];
\vertex  [left=0.5cm of k1] (k3) [label=above:\(i\)] ;
\vertex [below=0.5cm of k3] (k4);
\vertex [below=0.5cm of k2] (k5);
   \draw [thick,rounded corners=1mm] (k2)--(k5)--(i2);
      \draw [thick,rounded corners=1mm] (k3)--(k4)--(i2);
   \diagram*{
(i2)--[thick] (i1)
  };
  \end{feynman}
\end{tikzpicture}}}
\end{equation}
where $\alpha\in \{ 1,\ldots,\tensor{N}{_{ij}}{^k}\}$ with $N_{ij}^k$ being the fusion coefficients defined as 
    \begin{equation}
   N_{ij}^k:=\operatorname{dim}\operatorname{Hom}(i\times j,k) \, .
    \end{equation}
    Duality of these bases means that the following normalisation holds
\begin{equation} \label{bubble}
\vcenter{\hbox{\hspace{-10mm}\begin{tikzpicture}[font=\footnotesize,inner sep=2pt]
  \begin{feynman}
  \vertex  [small,orange, dot] (del) at (0,0) [label=below:\(\alpha\)] {} ;
  \vertex  [small,orange, dot] (gam) at (0,-2) [label=above:\(\beta\)] {} ;
 \vertex [below=1cm of gam] (j1) [label=below:\(j\)];
  \vertex [above=1cm of del] (j2) [label=above:\(m\)];
   \diagram*{
     (j2)--[thick] (del),
     (j1)--[thick] (gam),
     (del) --[half left, thick, looseness=1.2,edge label=$k$] (gam),
     (gam) --[half left, thick, looseness=1.2, edge label=$i$] (del)
  };
  \end{feynman}
\end{tikzpicture}}}
~\hspace{5mm}=\delta_{m,j}\delta_{\alpha,\beta} 
~ 
\vcenter{\hbox{\hspace{5mm}\begin{tikzpicture}[font=\footnotesize,inner sep=2pt]
  \begin{feynman}
  \vertex (1) at (0,1) [label=above:\(j\)];
  \vertex (2) at (0,-3) [label=below:\(j\)];
   \diagram*{
     (1)--[thick] (2)
  };
  \end{feynman}
\end{tikzpicture}}}\end{equation}
These bases also satisfy a completeness relation due to the semisimplicity of the category $\mathcal{C}$ \cite{FRS1}:
\begin{equation}\label{completeness}
    \vcenter{\hbox{\hspace{-10mm}\begin{tikzpicture}[font=\footnotesize,inner sep=2pt]
  \begin{feynman}
  \vertex  (i1) at (0,0) [label=below:\(i\)]  ;
  \vertex [above=3cm of i1] (i2) [label=above:\(i\)];
  \vertex [right=1cm of i1] (j1) [label=below:\(j\)];
  \vertex [right=1cm of i2] (j2) [label=above:\(j\)];
  \diagram*{
  (i1)--[thick] (i2),
  (j1)--[thick] (j2)
    };
  \end{feynman}
\end{tikzpicture}}}
~\hspace{5mm}=\mathlarger{\sum}_{k\in\mathcal{I}}\mathlarger{\sum}_\gamma\,
~ 
\vcenter{\hbox{\hspace{1mm}\begin{tikzpicture}[font=\footnotesize,inner sep=2pt]
  \begin{feynman}
  \vertex  [small,orange,dot] (up) at (0,0) [label=above:\(\gamma\)] {};
  \vertex [small,orange,dot] [below=1cm of up] (down) [label=below:\(\gamma\)] {};
  \vertex [above left=0.75cm of up] (i1);
  \vertex [above=0.6cm of i1] (i2) [label=above:\(i\)];
  \vertex [above right=0.75cm of up] (j1);
  \vertex [above=0.6cm of j1] (j2) [label=above:\(j\)];
  \draw [thick,rounded corners=1mm] (up) -- (i1) --(i2);
  \draw [thick,rounded corners=1mm] (up) -- (j1) --(j2);
  \vertex [below left=0.75cm of down] (i3) ;
  \vertex [below=0.6cm of i3] (i4) [label=below:\(i\)];
  \vertex [below right=0.75cm of down] (j3);
  \vertex [below=0.6cm of j3] (j4) [label=below:\(j\)];
  \draw [thick,rounded corners=1mm] (down) -- (j3) --(j4);
   \draw [thick,rounded corners=1mm] (down) -- (i3) --(i4);
   \diagram*{
      (down)--[thick,edge label=$k$] (up)
  };
  \end{feynman}
\end{tikzpicture}}}
\end{equation}
Another piece of notation needed is the braiding morphisms of $\mathcal{C}$ which will be depicted as 
\begin{equation}
     \vcenter{\hbox{\hspace{2mm}\begin{tikzpicture}[font=\footnotesize,inner sep=2pt]
  \begin{feynman}
  \vertex (j1) at (0,0) [label=below:\(i\)];
  \vertex [above=3cm of j1] (j2) [label=above:\(j\)];
  \vertex [above=1.7cm of j1] (j4);
  \draw [fill=yellow] (-0.2,1.7) rectangle (0.7, 1.3);
  \vertex [above=1.3cm of j1] (j3) [label=above:\( \hspace{4mm}c_{i,j}\)];
  \vertex [right=0.5cm of j1] (i1) [label=below:\(j\)];
  \vertex[above=3cm of i1] (i2) [label=above:\(i\)];
  \vertex[above=1.7 cm of i1] (i4);
  \vertex[ above=1.3 cm of i1] (i3);
   \diagram*{
    (j1)--[thick] (j3),
    (j4)--[thick] (j2),
    (i1)--[thick] (i3),
    (i4)--[thick] (i2)
  };
  \end{feynman}
\end{tikzpicture}}}
~=~
\vcenter{\hbox{\hspace{2mm}\begin{tikzpicture}[font=\footnotesize,inner sep=2pt]
  \begin{feynman}
  \vertex (j1) at (0,0) [label=below:\(i\)];
   \vertex [right=1cm of j1] (i1) [label=below:\(j\)];
  \vertex [above=3cm of j1] (j2) [label=above:\(j\)];
  \vertex [above=1cm of j1] (j4);
  \vertex [above=2cm of j1] (j5);
  \vertex [above=2cm of i1] (i5);
  \vertex[above=3cm of i1] (i2) [label=above:\(i\)];
  \vertex[above=1 cm of i1] (i4);
  \vertex[above left=0.6cm of i4] (i6);
   \vertex [above left=0.2 cm of i6] (i7);
   \draw [thick,rounded corners=1mm] (j4)--(i5)--(i2);
    \draw [thick,rounded corners=1mm] (i1)--(i4) -- (i6);
    \draw [thick,rounded corners=1mm] (i7)--(j5) -- (j2);
   \diagram*{
   (j1)--[thick] (j4)
  };
  \end{feynman}
\end{tikzpicture}}}
\hspace{15mm}
 \vcenter{\hbox{\hspace{2mm}\begin{tikzpicture}[font=\footnotesize,inner sep=2pt]
  \begin{feynman}
  \vertex (j1) at (0,0) [label=below:\(i\)];
  \vertex [above=3cm of j1] (j2) [label=above:\(j\)];
  \vertex [above=1.8cm of j1] (j4);
  \draw [fill=yellow] (-0.2,1.8) rectangle (0.7, 1.3);
  \vertex [above=1.3cm of j1] (j3) [label=above:\( \hspace{4mm}c^{-1}_{j,i}\)];
  \vertex [right=0.5cm of j1] (i1) [label=below:\(j\)];
  \vertex[above=3cm of i1] (i2) [label=above:\(i\)];
  \vertex[above=1.8 cm of i1] (i4);
  \vertex[ above=1.3 cm of i1] (i3);
   \diagram*{
    (j1)--[thick] (j3),
    (j4)--[thick] (j2),
    (i1)--[thick] (i3),
    (i4)--[thick] (i2)
  };
  \end{feynman}
\end{tikzpicture}}}
~=~
\vcenter{\hbox{\hspace{2mm}\begin{tikzpicture}[font=\footnotesize,inner sep=2pt]
  \begin{feynman}
  \vertex (j1) at (0,0) [label=below:\(i\)];
   \vertex [right=1cm of j1] (i1) [label=below:\(j\)];
  \vertex [above=3cm of j1] (j2) [label=above:\(j\)];
  \vertex [above=1cm of j1] (j4);
  \vertex [above=2cm of j1] (j5);
  \vertex [above=2cm of i1] (i5);
  \vertex[above=3cm of i1] (i2) [label=above:\(i\)];
  \vertex[above=1 cm of i1] (i4);
  \vertex[above right=0.6cm of j4] (j6);
   \vertex [above right=0.2 cm of j6] (j7);
   \draw [thick,rounded corners=1mm] (j1)--(j4)--(j6);
    \draw [thick,rounded corners=1mm] (j7)--(i5) -- (i2);
    \draw [thick,rounded corners=1mm] (i4)--(j5) -- (j2);
   \diagram*{
   (i1)--[thick] (i4)
  };
  \end{feynman}
\end{tikzpicture}}}
\end{equation}
Using the chosen basis $\{\lambda^\alpha_{(i,j)k}\}$ we define the  braiding matrices by
\begin{equation}\label{Bmatrix}
    \vcenter{\hbox{\hspace{2mm}\begin{tikzpicture}[font=\footnotesize,inner sep=2pt]
  \begin{feynman}
  \vertex (j1) at (0,0) [label=below:\(i\)];
   \vertex [right=1cm of j1] (i1) [label=below:\(j\)];
  \vertex [above=2cm of j1] (j2);
  \vertex [above=0.7cm of j1] (j4);
  \vertex [above=1.5cm of j1] (j5);
  \vertex [above=1.5cm of i1] (i5);
  \vertex[above=2cm of i1] (i2) ;
  \vertex[above=0.7 cm of i1] (i4);
  \vertex[above left=0.5cm of i4] (i6);
   \vertex [above left=0.2 cm of i6] (i7);
   \draw [thick,rounded corners=1mm] (j4)--(i5)--(i2);
    \draw [thick,rounded corners=1mm] (i1)--(i4) -- (i6);
    \draw [thick,rounded corners=1mm] (i7)--(j5) -- (j2);
    \vertex [right=0.5cm of j2] (k1);
    \vertex [small,orange, dot] [above=0.5cm  of k1] (k2) [label=right:\(\footnotesize \alpha\)] {};
    \draw [thick,rounded corners=1mm] (j2)--(k2);
     \draw [thick,rounded corners=1mm] (i2)--(k2);
     \vertex[above=0.6cm of k2] (k3) [label=above:\(k\)];
   \diagram*{
   (j1)--[thick] (j4),
   (k3)--[thick] (k2)
  };
  \end{feynman}
\end{tikzpicture}}}
~=\;\mathlarger{\sum}_\beta \mathrm{R}^{(i\,j)k}_{\alpha \beta}~
 \vcenter{\hbox{\begin{tikzpicture}[font=\footnotesize,inner sep=2pt]
  \begin{feynman}
\vertex (i1) at (0,0) [label=below:\(i\)];
\vertex [right=1cm of i1] (j1) [label=below:\(j\)];
\vertex [above=0.5cm of i1] (i2);
\vertex [above =0.5cm of j1] (j2);
\vertex  [right=0.5cm of i1] (k1) ;
\vertex [small,orange, dot][above =1.5cm of k1] (k2)[label=right:\(\beta\)] {};
\vertex [above=1.5cm of k2] (k3) [label=above:\(k\)];
   \draw [thick,rounded corners=1mm] (i1)--(i2)--(k2);
      \draw [thick,rounded corners=1mm] (j1)--(j2)--(k2);
   \diagram*{
(k2)--[thick] (k3)
  };
  \end{feynman}
\end{tikzpicture}}}
\end{equation}
as well as the fusing matrices by
\begin{equation}\label{Fmatrix}
   \vcenter{\hbox{\hspace{2mm}\begin{tikzpicture}[font=\footnotesize,inner sep=2pt]
   \begin{feynman}
  \vertex (l) at (0,0) [label=below:\(l\)] ;
  \vertex [small,orange,dot] [above=1cm of l] (gam) [label=above:\footnotesize\(\delta\)] {};
  \vertex [small,orange,dot] [above left=1.25cm of gam] (del) [label=above:\footnotesize\(\gamma\)] {} ;
  \vertex [above right=1cm of gam] (k1);
  \vertex [above=1.4cm of k1] (k) [label=above:\(k\)];
  \vertex [above left=0.7cm of del] (i1);
  \vertex [above=0.7cm of i1] (i) [label=above: \(i\)];
   \vertex [above right=0.7cm of del] (j1);
  \vertex [above=0.7cm of j1] (j) [label=above: \(j\)];
   \draw [thick,rounded corners=1mm] (gam) -- (k1) --(k);
     \draw [thick,rounded corners=1mm] (del) -- (j1) --(j);
       \draw [thick,rounded corners=1mm] (del) -- (i1) --(i);
    \diagram*{
     (l)--[thick] (gam),
     (gam) --[thick,edge label= $q$] (del)
    };
  \end{feynman}
\end{tikzpicture}}}
~=\;\mathlarger{\sum}_{p\in\mathcal{I}}\mathlarger{\sum}_{\alpha,\beta}\mathrm{F}^{(i\,j\,k)l}_{\alpha p\beta,\gamma q\delta}~
\vcenter{\hbox{\begin{tikzpicture}[font=\footnotesize,inner sep=2pt]
  \begin{feynman}
  \vertex (l) at (0,0) [label=below:\(l\)] ;
  \vertex [small,orange,dot] [above=1cm of l] (beta) [label=above:\footnotesize\(\alpha\)] {};
  \vertex [above left=1cm of beta] (i1);
  \vertex [above=1.4cm of i1] (i) [label=above:\(i\)];
  \vertex [small,orange,dot] [above right=1.25cm of beta] (alpha) [label=right:\footnotesize\(\,\beta\)] {};
  \vertex[above left=0.7cm of alpha] (j1);
  \vertex [above=0.7cm of j1] (j) [label=above:\(j\)];
  \vertex [above right=0.7cm of alpha] (k1);
    \vertex [above=0.7cm of k1] (k) [label=above:\(k\)];
    \draw [thick,rounded corners=1mm] (alpha) -- (j1) --(j);
    \draw [thick,rounded corners=1mm] (alpha) -- (k1) --(k);
    \draw [thick,rounded corners=1mm] (beta) -- (i1) --(i);
    \diagram*{
      (alpha) --[thick,edge label= $p$] (beta),
      (l)--[thick] (beta)
    };
  \end{feynman}
\end{tikzpicture}}}
\end{equation}
and their inverses by 
\begin{equation} \label{Gmatrix}
\vcenter{\hbox{\hspace{-10mm}\begin{tikzpicture}[font=\footnotesize,inner sep=2pt]
  \begin{feynman}
  \vertex (l) at (0,0) [label=below:\(l\)] ;
  \vertex [small,orange,dot] [above=1cm of l] (beta) [label=above:\(\beta\)] {};
  \vertex [above left=1cm of beta] (i1);
  \vertex [above=1.4cm of i1] (i) [label=above:\(i\)];
      \vertex [small,orange,dot] [above right=1.25cm of beta] (alpha) [label=right:\(\,\alpha\)]{};
  \vertex[above left=0.7cm of alpha] (j1);
  \vertex [above=0.7cm of j1] (j) [label=above:\(j\)];
  \vertex [above right=0.7cm of alpha] (k1);
    \vertex [above=0.7cm of k1] (k) [label=above:\(k\)];
    \draw [thick,rounded corners=1mm] (alpha) -- (j1) --(j);
    \draw [thick,rounded corners=1mm] (alpha) -- (k1) --(k);
    \draw [thick,rounded corners=1mm] (beta) -- (i1) --(i);
    \diagram*{
      (alpha) --[thick,edge label= $q$] (beta),
      (l)--[thick] (beta)
    };
  \end{feynman}
\end{tikzpicture}}}
~\hspace{5mm}=\mathlarger{\sum}_{p\in\mathcal{I}}\mathlarger{\sum}_{\gamma,\delta}\,\mathrm{G}_{\gamma p \delta,\alpha q\beta}^{(i\,j\,k)l}
~ 
\vcenter{\hbox{\hspace{2mm}\begin{tikzpicture}[font=\footnotesize,inner sep=2pt]
   \begin{feynman}
  \vertex (l) at (0,0) [label=below:\(l\)] ;
  \vertex [small,orange,dot] [above=1cm of l] (gam) [label=above:\(\gamma\)] {};
  \vertex [small,orange,dot] [above left=1.25cm of beta] (del) [label=above:\(\delta\)] {} ;
  \vertex [above right=1cm of gam] (k1);
  \vertex [above=1.4cm of k1] (k) [label=above:\(k\)];
  \vertex [above left=0.7cm of del] (i1);
  \vertex [above=0.7cm of i1] (i) [label=above: \(i\)];
   \vertex [above right=0.7cm of del] (j1);
  \vertex [above=0.7cm of j1] (j) [label=above: \(j\)];
   \draw [thick,rounded corners=1mm] (gam) -- (k1) --(k);
     \draw [thick,rounded corners=1mm] (del) -- (j1) --(j);
       \draw [thick,rounded corners=1mm] (del) -- (i1) --(i);
    \diagram*{
     (l)--[thick] (gam),
     (gam) --[thick,edge label=\small $p$] (del)
    };
  \end{feynman}
\end{tikzpicture}}}\end{equation}

Finally, in this framework the modular $S$-matrix of the CFT is defined as
\begin{equation}\label{sdef}
 \hspace{10mm}\mathlarger{S_{ij}}
~:=\mathlarger{S_{00}\operatorname{tr}\left(c_{i,j}\circ c_{j,i}\right)=}\; \mathlarger{S_{00}\dim(i)\dim(j)}
 \vcenter{\hspace{1mm}\hbox{\begin{tikzpicture}[font=\footnotesize,inner sep=2pt]
  \begin{feynman}
\vertex (i1) at (0,0) ;
\vertex [above=1cm of i1] (i2);
\vertex[right=2cm of i1] (i3);
\vertex[right=2cm of i2] (i4);
\draw  [thick] (i4)-- (i3);
\draw [thick, half left, looseness=1.5] (i2) to  (i4);
\vertex[right=0.64cm of i1] (j1);
\vertex[above=0.5cm of j1] (j11) [label=right:\(\hspace{1mm}j\)];
\vertex[above=0.5cm of i1] (i11) ;
\vertex[left=0.1cm of i11] (i111) [label=left:\(i\)];
\vertex[above=1cm of j1] (j2);
\vertex[left=2cm of j1] (j4);
\vertex[left=2cm of j2] (j3);
\vertex[below right=0.5cm of i1] (m1);
\vertex[right=0.25cm of m1] (m11);
\vertex[below=0.15cm of m11] (m2);
\vertex[right=0.4cm of m2] (m3);
\vertex[below=0.35cm of m3] (botr) ;
\draw [thick, quarter right, looseness=1.2] (botr) to (i3);
\draw [thick]    (m2) to[out=-45,in=-30] (botr);
\draw [thick]    (i1) to[out=-90,in=-20] (m1);
\vertex[above left=0.5cm of j2] (n11);
\vertex[above=0.1cm of n11] (n1);
\vertex[left=0.25cm of n1] (n2);
\vertex[above=0.2cm of n2] (n3);
\vertex[left=0.4cm of n3] (n4);
\vertex[small,orange,dot][above=0.2cm of n4] (topl){};
\vertex[above=0.6cm of topl] (topl0) [label=above:\(0\)];
\draw [thick]    (n3) to[out=-30,in=-25] (topl);
\draw [thick]    (j2) to[out=90,in=-45] (n1);
\draw [thick, half left, looseness=1.5] (j1) to  (j4);
\vertex[small,orange,dot] [below=2.78cm of topl] (botl){};
\vertex[below=0.6cm of botl] (botl0) [label=below:\(0\)];
\vertex[small,orange,dot][right=0.0001cm of botr] (botrr){};
\vertex[below=0.62cm of botrr] (botr0) [label=below:\(0\)];
\vertex[small,orange,dot][above=2.68cm of botr] (topr) {};
\vertex[above=0.62cm of topr] (topr0) [label=above:\(0\)];
   \diagram*{
(i1)--[thick] (i2),
(j1)--[thick] (j2),
(j3)--[thick] (j4),
(topl) -- [thick, quarter right, looseness=1.2] (j3),
(topl)--[dashed] (topl0),
(botl)--[dashed] (botl0),
(botrr)--[dashed] (botr0),
(topr0)--[dashed] (topr)
  };
  \end{feynman}
\end{tikzpicture}}}
\end{equation}
where $\dim(i)$ is the quantum dimension of an object $U_i\in\operatorname{Obj}\left(\mathcal{C}\right)$. It is expressed in terms of the $S$-matix as
\begin{equation}
    \operatorname{dim}(i)= \frac{S_{i0}}{S_{00}}\,.
\end{equation}

We will work with the trivial Frobenius algebra $A=\mathbb{1}$ which corresponds to the charge conjugation modular invariant theory.  In this case the topological defects are in one-to-one corerspondence with the simple objects of $\mathcal{C}$ hence labelled by the same set $\mathcal{I}$ as  the irreducible representations of the chiral algebra. Using the representations of field insertions introduced in \cite{FRS4}, in Figure \ref{figintro1} the morphism representing the bulk field insertion is
\begin{equation}
    \phi_{i,\overline{i}}\in\operatorname{Hom}( i\times \overline{i},\mathbb{1}) %% changed 0 to 1
\end{equation}
while for the defect field 
\begin{equation}
    \mu_{i,\overline{i}}^{e;\beta}\in \Hom_{\mathbb{1}|\mathbb{1}}\left(i\times \mathbb{1}\times  \overline{i},e\right) \cong \Hom(i\times \overline{i},e)
\end{equation}
so we can identify the basis $ \mu_{i,\overline{i}}^{e;\beta}$ of defect fields  with the  basis $\lambda_{(i,\overline{i})e}^\beta$. Finally, the triple defect junction appearing in Figure \ref{figintro1} is represented by $\lambda^\alpha_{(d,e)d}$. Using this notation we rewrite Figure \ref{figintro1} as the following equation between morphisms

\begin{equation}\label{Qcalc1}
    \vcenter{\hbox{\hspace{-10mm}\begin{tikzpicture}[font=\footnotesize,inner sep=2pt]
  \begin{feynman}
  \vertex  (i1) at (0,0) [label=below:\(i\)]  ;
  \vertex [above=1.2cm of i1] (i2) ;
  \vertex [right=1cm of i1] (d1) [label=below:\(d\)];
  \vertex[above=4cm of d1] (d4) [label=above:\(d\)];
    \vertex[above=1.6cm of d1] (d2);
    \vertex[above=2cm of d1] (d3);
    \vertex[right=1cm of d1] (m0);
    \vertex [small,orange,dot][above=2.35cm of m0] (m1) {};
    \draw [thick,rounded corners=1mm] (i1) -- (i2) --(m1);
    \vertex[right=1cm of m0] (j1) [label=below:\(\overline{i}\)] ;
    \vertex[above=1.2cm of j1] (j2);
    \draw [thick,rounded corners=1mm] (j1) -- (j2) --(m1);
    \vertex[right=1cm of d4] (m2) [label=above:\(0\)];
    \draw [dashed] (m1) -- (m2);
  \diagram*{
  (d1)--[thick] (d2),
  (d3)--[thick] (d4),
  
    };
  \end{feynman}
\end{tikzpicture}}}
~\hspace{5mm}=\mathlarger{\sum}_{e\in\mathcal{I}}\mathlarger{\sum}_{\alpha,\beta}\, Q^{\alpha e\beta}_{id}
~ 
\vcenter{\hbox{\hspace{5mm}\begin{tikzpicture}[font=\footnotesize,inner sep=2pt]
  \begin{feynman}
   \vertex  (i1) at (0,0) [label=below:\(i\)]  ;
  \vertex [above=1cm of i1] (i2) ;
  \vertex (i3) at (0.9,1.6);
  \vertex (i4) at (1.1,1.7);
  \vertex [right=1cm of i1] (d1) [label=below:\(d\)];
  \vertex[above=4cm of d1] (d4) [label=above:\(d\)];
    \vertex[above=1.6cm of d1] (d2);
    \vertex[above=2cm of d1] (d3);
    \vertex[right=1cm of d1] (m0);
    \vertex [small,orange,dot][above=2.35cm of m0] (m1) [label=right:\(\beta\)] {};
    \draw [thick,rounded corners=1mm] (i1) -- (i2) --(i3);
    \vertex[right=1cm of m0] (j1) [label=below:\(\overline{i}\)] ;
    \vertex[above=1cm of j1] (j2);
    \draw [thick,rounded corners=1mm] (j1) -- (j2) --(m1);
    \vertex[right=1cm of d4] (m2) ;
    \vertex[above=0.7cm of m1] (m3);
    \vertex (e) at (1.8,2.8) [label=\(e\)];
    \draw [thick] (d1)--(d4);
    \vertex[small,orange,dot][below=0.5cm of d4] [label=left:\(\alpha\)] (m4){};
     \draw [thick,rounded corners=1mm] (m1) -- (m3) --(m4);
  \diagram*{
  (i4)--[thick] (m1)
    };
  \end{feynman}
\end{tikzpicture}}}
\end{equation}
To extract the coefficients $Q$ we compose \eqref{Qcalc1} with the following dual morphism
\begin{equation}\label{Qcalc2}
    \vcenter{\hbox{\hspace{-10mm}\begin{tikzpicture}[font=\footnotesize,inner sep=2pt]
  \begin{feynman}
  \vertex  (i1) at (0,0) [label=above:\(i\)]  ;
  \vertex [right=1cm of i1] (d1) [label=above:\(d\)] ;
  \vertex[right=3cm of i1] (j1)[label=above:\(\overline{i}\)];
 \vertex[right=2cm of i1] (m1);
 \vertex[below=4cm of d1] (d4) [label=below:\(d\)];
 \draw [thick] (d1)-- (d4);
 \vertex[below=1cm of i1] (i2);
 \vertex[below=1cm of j1] (j2);
 \vertex[small,orange,dot][below=2.35cm of m1] (m2) [label=above:\(\;\beta^\prime\)]{};
 \vertex (i3) at (0.9,-1.6);
 \draw [thick,rounded corners=1mm] (i1) -- (i2) --(i3);
  \draw [thick,rounded corners=1mm] (j1) -- (j2) --(m2);
 \vertex (i4) at (1.1,-1.7);
 \vertex[below=0.7cm of m2] (m3);
 \vertex[small,orange,dot][above=0.5cm of d4] (m4) [label=left:\(\alpha^\prime\)]{};
 \draw [thick,rounded corners=1mm] (m2) -- (m3) --(m4);
 \vertex (e) at (1.8,-3.1) [label=\(e^\prime\)];
  \diagram*{
 (i4)--[thick] (m2)
    };
  \end{feynman}
\end{tikzpicture}}}
\end{equation}
and we find
\begin{equation}\label{Qcalc3}
    \vcenter{\hbox{\hspace{-10mm}\begin{tikzpicture}[font=\footnotesize,inner sep=2pt]
  \begin{feynman}
  \vertex  (i1) at (0,0)   ;
  \vertex[left=0.1cm of i1] (i11) ;
  \vertex[above=0.6cm of i11] (i111) [label=left:\(i\)];
  \vertex [above=1.2cm of i1] (i2) ;
  \vertex [right=1cm of i1] (d1) ;
  \vertex[below=2.5cm of d1] (dbot) [label=below:\(d\)];
  \vertex[above=3cm of d1] (d4) [label=above:\(d\)];
    \vertex[above=1.6cm of d1] (d2);
    \vertex[above=2cm of d1] (d3);
    \vertex[right=1cm of d1] (m0);
    \vertex [small,orange,dot][above=2.35cm of m0] (m1) {};
    \draw [thick,rounded corners=1mm] (i1) -- (i2) --(m1);
    \vertex[right=1cm of m0] (j1)  ;
    \vertex[right=0.1cm of j1] (j11) ;
    \vertex[above=0.6cm of j11] (j111) [label=right:\(\overline{i}\)];
    \vertex[above=1.2cm of j1] (j2);
    %\draw [thick,rounded corners=1mm] (j1) -- (j2) --(m1);
    \vertex[right=1cm of d4] (m2) [label=above:\(0\)];
    \draw [dashed] (m1) -- (m2);
 \vertex (i3) at (0.85,-0.5);
 \vertex (i4) at (1.1,-0.6);
 \vertex[small,orange,dot][below=3.5cm of m1] (m3)[label=above:\(\beta\)] {};
 \vertex[below=0.6cm of m3] (e1);
\draw [thick] (dbot)-- (d1);
  \vertex[small,orange,dot][above=0.5cm of dbot] (e2) [label=left:\(\alpha\)]{};
  \draw [thick,rounded corners=1mm] (m3) -- (e1) --(e2);
  \vertex (e3) at (1.8,-1.7) [label=\(e\)];
  \diagram*{
  (d1)--[thick] (d2),
  (d3)--[thick] (d4),
  (i3)--[thick] (i1),
  (m3)--[thick] (i4),
  (m3) -- [thick, half right, looseness=1] (m1)
    };
  \end{feynman}
\end{tikzpicture}}}
~\hspace{5mm}= Q^{\alpha e\beta}_{id}
~ 
\vcenter{\hbox{\hspace{5mm}\begin{tikzpicture}[font=\footnotesize,inner sep=2pt]
  \begin{feynman}
   \vertex  (i1) at (0,0) [label=below:\(d\)]  ;
  \vertex [above=5cm of i1] (i2) [label=above:\(d\)] ;
  \diagram*{
  (i1)--[thick] (i2)
    };
  \end{feynman}
\end{tikzpicture}}}
\end{equation}
from which we can obtain $Q^{\alpha e \beta}_{id}$ by manipulating the diagram on the left-hand side. More specifically we find that the left-hand side of \eqref{Qcalc3} is
\begin{align}\label{Qcalc4}
   & \vcenter{\hbox{\hspace{-10mm}\begin{tikzpicture}[font=\footnotesize,inner sep=2pt]
  \begin{feynman}
  \vertex  (i1) at (0,0)   ;
  \vertex[left=0.1cm of i1] (i11) ;
  \vertex[above=0.6cm of i11] (i111) [label=left:\(i\)];
  \vertex [above=1.2cm of i1] (i2) ;
  \vertex [right=1cm of i1] (d1) ;
  \vertex[below=2.5cm of d1] (dbot) [label=below:\(d\)];
  \vertex[above=3cm of d1] (d4) [label=above:\(d\)];
    \vertex[above=1.6cm of d1] (d2);
    \vertex[above=2cm of d1] (d3);
    \vertex[right=1cm of d1] (m0);
    \vertex [small,orange,dot][above=2.35cm of m0] (m1) {};
    \draw [thick,rounded corners=1mm] (i1) -- (i2) --(m1);
    \vertex[right=1cm of m0] (j1)  ;
    \vertex[right=0.1cm of j1] (j11) ;
    \vertex[above=0.6cm of j11] (j111) [label=right:\(\overline{i}\)];
    \vertex[above=1.2cm of j1] (j2);
    %\draw [thick,rounded corners=1mm] (j1) -- (j2) --(m1);
    \vertex[right=1cm of d4] (m2) [label=above:\(0\)];
    \draw [dashed] (m1) -- (m2);
 \vertex (i3) at (0.85,-0.5);
 \vertex (i4) at (1.1,-0.6);
 \vertex[small,orange,dot][below=3.5cm of m1] (m3)[label=above:\(\beta\)] {};
 \vertex[below=0.6cm of m3] (e1);
\draw [thick] (dbot)-- (d1);
  \vertex[small,orange,dot][above=0.5cm of dbot] (e2) [label=left:\(\alpha\)]{};
  \draw [thick,rounded corners=1mm] (m3) -- (e1) --(e2);
  \vertex (e3) at (1.8,-1.7) [label=\(e\)];
  \diagram*{
  (d1)--[thick] (d2),
  (d3)--[thick] (d4),
  (i3)--[thick] (i1),
  (m3)--[thick] (i4),
  (m3) -- [thick, half right, looseness=1] (m1)
    };
  \end{feynman}
\end{tikzpicture}}}
~= \mathlarger{\sum_{j,\rho}}
~ 
 \vcenter{\hbox{\hspace{2mm}\begin{tikzpicture}[font=\footnotesize,inner sep=2pt]
  \begin{feynman}
  \vertex  (i1) at (0,0)   ;
  \vertex[left=0.1cm of i1] (i11) ;
  \vertex[above=0.6cm of i11] (i111) ;
  \vertex [above=1.2cm of i1] (i2) ;
  \vertex [right=1cm of i1] (d1) ;
  \vertex[below=2.5cm of d1] (dbot) [label=below:\(d\)];
  \vertex[above=3cm of d1] (d4) [label=above:\(d\)];
    \vertex[above=1.7cm of d1] (d2);
     \draw [thick] (dbot)-- (d2);
    \vertex[above=2cm of d1] (d3);
    \vertex[right=1cm of d1] (m0);
    \vertex [small,orange,dot][above=2.35cm of m0] (m1) {};
    %\draw [thick,rounded corners=1mm] (i1) -- (i2) --(m1);
    \vertex[right=1cm of m0] (j1)  ;
    \vertex[right=0.1cm of j1] (j11) ;
    \vertex[above=0.6cm of j11] (j111) [label=right:\(\overline{i}\)];
    \vertex[above=1.2cm of j1] (j2);
    %\draw [thick,rounded corners=1mm] (j1) -- (j2) --(m1);
    \vertex[right=1cm of d4] (m2) [label=above:\(0\)];
    \draw [dashed] (m1) -- (m2);
 \vertex (i3) at (0.85,-0.5);
 \vertex (i4) at (1.1,-0.6);
 \vertex[small,orange,dot][below=3.5cm of m1] (m3)[label=above:\(\beta\)] {};
 \vertex[below=0.6cm of m3] (e1);
  \vertex[small,orange,dot][above=0.5cm of dbot] (e2) [label=left:\(\alpha\)]{};
  \draw [thick,rounded corners=1mm] (m3) -- (e1) --(e2);
  \vertex (e3) at (1.8,-1.7) [label=\(e\)];
  \vertex[above right=0.5cm of i2] (a1)[label=left:\(i\)];
  \vertex[below=0.4cm of a1] (a2);
  \vertex[small,orange,dot][below=2cm of d4] (a3)[label=right:\(\rho\)]{};
    \draw [thick,rounded corners=1mm] (a3)--(a2) -- (a1) --(m1);
    \vertex (a4) at (0.4,0.1);
    \vertex[below=0.3cm of a4] (a6) [label=left:\(i\)];
    \vertex[small,orange,dot][below=0.6cm of a3] (a5) [label=right:\(\rho\)] {};
     \draw [thick,rounded corners=1mm] (a5) -- (a4) --(a6)--(i3);
     \vertex [below=0.3cm of a3] (b1)[label=left:\(j\)];
  \diagram*{
  (d3)--[thick] (d4),
  (m3)--[thick] (i4),
  (m3) -- [thick, half right, looseness=1] (m1)
    };
  \end{feynman}
\end{tikzpicture}}} 
~=~
\mathlarger{\sum_{j,\rho,\mu,\nu}}\mathrm{R}^{(id)j}_{\mu\rho}\,\mathrm{R}^{(di)j}_{\rho\nu}
  \vcenter{\hbox{\hspace{5mm}\begin{tikzpicture}[font=\footnotesize,inner sep=2pt]
  \begin{feynman}
  \vertex  (i1) at (0,0)   ;
  \vertex [above=1.2cm of i1] (i2) ;
  \vertex [right=1cm of i1] (d1) ;
  \vertex[below=2.5cm of d1] (dbot) [label=below:\(d\)];
  \draw [thick] (dbot)--(d4);
  \vertex[above=3cm of d1] (d4) [label=above:\(d\)];
    \vertex[above=1.6cm of d1] (d2);
    \vertex[above=2cm of d1] (d3);
    \vertex[right=1cm of d1] (m0);
    \vertex [small,orange,dot][above=2.35cm of m0] (m1) {};
    \vertex[right=1cm of m0] (j1)  ;
    \vertex[right=0.1cm of j1] (j11) ;
    \vertex[above=0.6cm of j11] (j111) [label=right:\(\overline{i}\)];
    \vertex[above=1.2cm of j1] (j2);
    %\draw [thick,rounded corners=1mm] (j1) -- (j2) --(m1);
    \vertex[right=1cm of d4] (m2) [label=above:\(0\)];
    \draw [dashed] (m1) -- (m2);
 \vertex (i3) at (0.85,-0.5);
 \vertex (i4) at (1.1,-0.6);
 \vertex[small,orange,dot][below=3.5cm of m1] (m3)[label=above:\(\beta\)] {};
 \vertex[below=0.6cm of m3] (e1);
  \vertex[small,orange,dot][above=0.5cm of dbot] (e2) [label=left:\(\alpha\)]{};
  \draw [thick,rounded corners=1mm] (m3) -- (e1) --(e2);
  \vertex (e3) at (1.8,-1.7) [label=\(e\)];
  \vertex[small,orange,dot][below=1.2cm of d4] (c1) [label=left:\(\mu\)]{};
  \vertex (c2) at (1.5,1.8) [label=\(i\)];
  \vertex (c4) at (1.5,-0.8) [label=\(i\)];
  \vertex[below=1cm of c1] (c5) [label=left:\(j\)];
  \vertex[small,orange,dot][below=2cm of c1] (c3) [label=left:\(\nu\)] {};
  \vertex[below=0.9cm of c3] (c6) [label=left:\(d\)];
  \diagram*{
  (m3) -- [thick, half right, looseness=1] (m1),
  (c1) -- [thick, quarter left, looseness=0.4] (m1),
  (c3)--[thick, quarter right, looseness=0.4] (m3)
    };
  \end{feynman}
\end{tikzpicture}}} 
\nonumber \\[2em]
 &= ~\mathlarger{\sum_{j,\rho,\mu,\nu,\gamma}}\mathrm{R}^{(id)j}_{\mu\rho}\,\mathrm{R}^{(di)j}_{\rho\nu} \,\mathrm{G}^{(di\overline{i})d}_{\gamma j\nu,\beta e\alpha}
 \vcenter{\hbox{\hspace{5mm}\begin{tikzpicture}[font=\footnotesize,inner sep=2pt]
  \begin{feynman}
  \vertex  (i1) at (0,0)   ;
  \vertex [above=1.2cm of i1] (i2) ;
  \vertex [right=1cm of i1] (d1) ;
  \vertex[below=1.5cm of d1] (dbot) [label=below:\(d\)];
  \draw [thick] (dbot)--(d4);
  \vertex[above=3cm of d1] (d4) [label=above:\(d\)];
    \vertex[above=1.6cm of d1] (d2);
    \vertex[above=2cm of d1] (d3);
    \vertex[right=1cm of d1] (m0);
    \vertex [small,orange,dot][above=2.35cm of m0] (m1) {};
    \vertex[right=1cm of m0] (j1)  ;
    \vertex[right=0.1cm of j1] (j11) ;
    \vertex[above=0.6cm of j11] (j111);
    \vertex[above=1.2cm of j1] (j2);
    %\draw [thick,rounded corners=1mm] (j1) -- (j2) --(m1);
    \vertex[right=1cm of d4] (m2) [label=above:\(0\)];
    \draw [dashed] (m1) -- (m2);
 \vertex (i3) at (0.85,-0.5);
 \vertex (i4) at (1.1,-0.6);
 \vertex[below=0.6cm of m3] (e1);
  \vertex[small,orange,dot][below=1.2cm of d4] (c1) [label=left:\(\mu\)]{};
  \vertex (c2) at (1.5,1.8) [label=\(i\)];
  \vertex[below=1cm of c1] (c5) [label=left:\(j\)];
  \vertex[small,orange,dot][below=2cm of c1] (c3) [label=left:\(\gamma\)] {};
  \vertex[below=1.5cm of m1] (r1);
  \draw [thick,rounded corners=1mm] (c3) -- (r1) --(m1);
  \vertex (r2) at (1.6,-0.1)  [label=:\(\overline{i}\)];
  \diagram*{
  (c1) -- [thick, quarter left, looseness=0.4] (m1)
    };
  \end{feynman}
\end{tikzpicture}}} 
~=\mathlarger{\sum_{j,\rho,\mu,\nu,\gamma}}\mathrm{R}^{(id)j}_{\mu\rho}\,\mathrm{R}^{(di)j}_{\rho\nu} \,\mathrm{G}^{(di\overline{i})d}_{\gamma j\nu,\beta e\alpha}\,\mathrm{F}^{(di\overline{i})d}_{0,\mu j\gamma}
\vcenter{\hbox{\hspace{5mm}\begin{tikzpicture}[font=\footnotesize,inner sep=2pt]
  \begin{feynman}
  \vertex  (i1) at (0,0) [label=above:\(d\)]   ;
  \vertex [below=4.4cm of i1] (i2) [label=below:\(d\)] ;
  \diagram*{
  (i1) -- [thick] (i2)
    };
  \end{feynman}
\end{tikzpicture}}} 
\end{align}
where at the first step we used the completeness property \eqref{completeness}, at the second step the definition of the braiding matrices \eqref{Bmatrix}, at the third step the $\mathrm{G}$-matrix move \eqref{Gmatrix} involving the vertices $\alpha,\beta$ and at the final step the $\mathrm{F}$-matrix move \eqref{Fmatrix} involving $\gamma,\mu$. Now substituting \eqref{Qcalc4} into \eqref{Qcalc3} we obtain
\begin{equation}\label{Qgeneralappendix}
      Q^{\alpha e\beta}_{id}=\sum_{j,\rho,\mu,\nu,\gamma}\mathrm{R}^{(id)j}_{\mu\rho}\, \mathrm{R}^{(di)j}_{\rho\nu}\, \mathrm{G}^{(di\overline{i})d}_{\gamma j \nu,\beta e \alpha}\,\mathrm{F}^{(di\overline{i})d}_{0,\mu j\gamma}\,.
\end{equation}

The coefficient $Q^0_{id}$ which contains information about the commutation or anticommutation property between the defect $X_d$ and the bulk field $\phi_{i,\overline{i}}$ can be alternatively obtained in terms of the $S$-matrix by composing \eqref{Qcalc1} with \eqref{Qcalc2} for $e^\prime=0$  and then taking the trace:
\begin{equation}
      \mathlarger{}\vcenter{\hbox{\hspace{-10mm}\begin{tikzpicture}[font=\footnotesize,inner sep=2pt]
  \begin{feynman}
  \vertex  (i1) at (0,0)   ;
  \vertex[left=0.1cm of i1] (i11) ;
  \vertex[above=0.6cm of i11] (i111) [label=right:\(\;\,i\)];
  \vertex [above=1.2cm of i1] (i2) ;
  \vertex [right=1cm of i1] (d1) ;
  \vertex[small,orange,dot][below=1.5cm of d1] (dbot){} ;
  \vertex[below=0.75cm of dbot] (dbot2)[label=below:\(0\)];
  \vertex[small,orange,dot][above=3cm of d1] (d4){} ;
  \vertex[above=0.55cm of d4] (d5) [label=above:\(0\)];
    \vertex[above=1.6cm of d1] (d2);
    \vertex[above=2cm of d1] (d3);
    \vertex[right=1cm of d1] (m0);
    \vertex [small,orange,dot][above=2.35cm of m0] (m1) {};
    \draw [thick,rounded corners=1mm] (i1) -- (i2) --(m1);
    \vertex[right=1cm of m0] (j1)  ;
    \vertex[right=0.1cm of j1] (j11) ;
    \vertex[above=0.6cm of j11] (j111) [label=right:\(\overline{i}\)];
    \vertex[above=1.2cm of j1] (j2);
    %\draw [thick,rounded corners=1mm] (j1) -- (j2) --(m1);
    \vertex[above=1.2cm of m1] (m2) [label=above:\(0\)];
    \draw [dashed] (m1) -- (m2);
 \vertex (i3) at (0.85,-0.5);
 \vertex (i4) at (1.1,-0.6);
 \vertex[small,orange,dot][below=3.5cm of m1] (m3){};
 \vertex[below=1.2cm of m3] (m4)[label=below:\(0\)];
 \vertex[below=0.6cm of m3] (e1);
 \vertex (d) at (1.2,0.4) [label=:\(d\)];
\draw [thick] (dbot)-- (d1);
  \diagram*{
  (d1)--[thick] (d2),
  (d3)--[thick] (d4),
  (i3)--[thick] (i1),
  (m3)--[thick] (i4),
  (m3) -- [thick, half right, looseness=1] (m1),
  (d4) -- [thick, half right, looseness=1] (dbot),
  (m3)--[dashed] (m4),
  (dbot)--[dashed] (dbot2),
  (d4)--[dashed] (d5)
    };
  \end{feynman}
\end{tikzpicture}}}
~\hspace{5mm}= Q^{0}_{id}
~ 
\vcenter{\hbox{\hspace{5mm}\begin{tikzpicture}[font=\footnotesize,inner sep=2pt]
  \begin{feynman}
   \vertex[small,orange,dot]  (i1) at (0,0) {}  ;
  \vertex [small,orange,dot][above=5cm of i1] (i2) {} ;
  \vertex[above=0.5cm of i2] (top) [label=above:\(0\)];
  \vertex[below=0.5cm of i1] (bot)[label=below:\(0\)];
  \draw [thick] (i1)--(i2);
  \vertex (d) at (-0.2,2.5) [label=:\(d\)];
  \vertex[small,orange,dot][below=0.7cm of i2] (i3){};
  \vertex[below=0.6cm of i3] (i4);
  \vertex[small,orange,dot][right=1cm of i4] (r1){};
  \vertex[small,orange,dot][below=2cm of r1] (r2){};
  \vertex[small,orange,dot][above=0.7cm of i1] (i5){};
  \vertex (i) at (1.9,2.5) [label=\(i\)];
  \diagram*{
   (i1) -- [thick, half left, looseness=1] (i2),
   (r1)--[dashed] (i3),
    (r1) -- [thick, half left, looseness=1] (r2),
     (r1) -- [thick, half right, looseness=1] (r2),
     (i5)--[dashed] (r2),
     (i1)--[dashed] (bot),
     (i2)--[dashed] (top)
    };
  \end{feynman}
\end{tikzpicture}}}
\end{equation}
In the left-hand side of the last equation we can now  use the $S$-matrix definition \eqref{sdef}, while on the right-hand side we can collapse the bubbles using \eqref{bubble}. In the end we obtain
\begin{equation}\label{Q0coeff}
    Q^0_{id}=\frac{S_{id}}{S_{00}\operatorname{dim}(i)\operatorname{dim}(d)}\,.
\end{equation}
%\begin{remark}
    Notice that using \eqref{Qgeneralappendix} for $e=0$ and substituting in \eqref{Q0coeff} we obtain an expression for the $S$-matrix  in terms of the braiding and fusing matrices:
    \begin{equation}\label{smatrixapp1}
        S_{id}=S_{00}\operatorname{dim}(i)\dim(d)\sum_{j,\rho,\mu,\nu,\gamma} \mathrm{R}^{(id)j}_{\mu\rho}\, \mathrm{R}^{(di)j}_{\rho\nu}\, \mathrm{G}^{(di\overline{i})d}_{\gamma j \nu,0}\,\mathrm{F}^{(di\overline{i})d}_{0,\mu j\gamma}\,.
    \end{equation}
    For the case that the dimensions $N_{ij}^k$ of all hom-spaces $\operatorname{Hom}(i\times j,k)$ are either $0$ or $1$ \eqref{smatrixapp1} yields
    \begin{equation}\label{smatrixapp2}
        S_{id}=S_{00}\operatorname{dim}(i)\dim(d)\sum_{j\in\mathcal{I}} \mathrm{R}^{(id)j}\, \mathrm{R}^{(di)j}\, \mathrm{G}^{(di\overline{i})d}_{j,0}\,\mathrm{F}^{(di\overline{i})d}_{0, j}\,.
    \end{equation}
    An equivalent expression was obtained in equation (2.64) in \cite{FRS1}. This is not an identical expression to our \eqref{smatrixapp2} as we used a different choice of $\mathrm{F}$ and $\mathrm{G}$ moves to evaluate the diagram on the left-hand side of \eqref{Qcalc3}.  The two expressions coincide after using  equation (2.43) from \cite{FRS1}.
%\end{remark}

\section{Details for \texorpdfstring{$\aff$}{affine \mathrm{SU(3)}} statements } \label{appendix1}

In this appendix we give a proof of lemma \ref{lemma2} and present details of proof of theorem \ref{thmsu3}. We start with the lemma.
%We present the proof of lemma \ref{lemma2} that was used in section \eqref{sec4.4}.
\begin{replemma}{lemma2}
    Let $T_1,T_2,T_3 \subseteq \mathcal{R}_{(a,b),(b,a)}^k$ be the off-diagonal representations that appear in the direct product decomposition of the  $\aff$ fusion $(a,b)\times (b,a)$. Then 
   \begin{enumerate}[i),wide, labelwidth=!, labelindent=0pt]\item \begin{equation}
        \{(k,0),(0,k)\}  \nsubseteq T_1\cup T_3
    \end{equation}
    and \item
    \begin{equation}
          \{(k,0),(0,k)\}  \subseteq T_2 \Longleftrightarrow a=b \quad \text{and}\quad k=3a
    \end{equation}
    \end{enumerate}
\end{replemma}
\begin{proof}
    It suffices to work with only one simple current say $(0,k)$ since they either both appear in $(a,b)\times (b,a)$ or none does. 
    \begin{enumerate}[i),wide, labelwidth=!, labelindent=0pt]
\item Suppose that $\{(0,k)\}\subseteq T_1$. Recalling the parametrisation of the off-diagonal representations, we have that $(0,k)$ is of the form $(p-2l,p+l)$ when $p=\frac{2}{3}k$ and $l=\frac{1}{3}k$. Then, since $\{(0,k)\}\subseteq T_1$, the following holds
\begin{equation}
    a+2 \leq \frac{2}{3}k \leq  \operatorname{min}\left(a+b,k-a-1,\floor*{\frac{2k}{3}},\floor*{\frac{b+k}{2}}\right)
\end{equation}
which gives the system
 \begin{equation}
         \left\{\begin{array}{lr}
     a+2\leq \frac{2}{3}k \\
        \frac{2}{3}k\leq a+b \\
        \frac{2}{3}k \leq k-a-1 \\
         \frac{2}{3}k\leq \floor*{\frac{b+k}{2}}
        \end{array}\right. 
 \Longrightarrow    \left\{\begin{array}{lr}
     2k\geq 3a+6 \\
        2k\leq 3a+3b \\
        2k \geq  6a+6 \\
         2k\leq 6b
        \end{array}\right. 
    \end{equation}
    Since $a\geq b$ the last two inequalities give a contradiction, hence $\{(0,k)\}\nsubseteq T_1$.
    
    Now suppose  $\{(0,k)\}\subseteq T_3$, then the following holds
    \begin{equation}
        2\leq \frac{2}{3}k \leq \operatorname{min}(k-a-1,2b-2)
    \end{equation}
    which gives the system
    \begin{equation}
         \left\{\begin{array}{lr}
     2k\geq 6 \\
        k \geq  3a+3\\
        k \leq 3b-3 \\
        \end{array}\right. 
    \end{equation}
    Since $a\geq b$ the last two inequalities give a contradiction, hence  $\{(0,k)\}\nsubseteq T_3$.
    \item 
    $(\Rightarrow)$: Suppose $\{(0,k)\}\subset T_2$. This implies that
    \begin{equation}
        b+1 \leq \frac{2}{3}k \leq \operatorname{min}(a+b,k-a)
    \end{equation}
    is true for some values of $k$. Thus we obtain the system
    \begin{equation}\label{lemma1}
         \left\{\begin{array}{lr}
     2k\geq 3b+3\\
        2k \leq  3a+3b\\
        2k \geq 6a \\
        \end{array}\right. 
        \Longrightarrow  \left\{\begin{array}{lr}
        2k \leq  3a+3b\\
        2k \geq 6a \\
        \end{array}\right.
    \end{equation}
    which must be satisfied for some values of $k$. For this to happen, since $a\geq b$ we must have $a=b$ in which case there is a single level $k=3a$ that satisfies \eqref{lemma1}. Finally, for $a=b$ and $k=3a$ we have that the value of $l$ which gives the simple current $(0,k)$ is $l=a$. It is straightforward to verify that this value appears in the $l$ interval for $T_2$ specified in \eqref{offdiag3} without imposing any further condition involving $k$ or $a$.
    
    $(\Leftarrow)$ Straightforward to check this by substituting $a=b$ and $k=3a$ in the definition \eqref{offdiag3} of $T_2$.
    \end{enumerate}
\end{proof}

We next turn to details of proof of theorem \ref{thmsu3}. We list the statements that we have proven in order to obtain tables \ref{table1}, \ref{table2} and \ref{table3}. We include some of the proofs in order to showcase the procedure.
\begin{description}[wide, labelwidth=!, labelindent=0pt]
\item{Statement 1:}  For all $(i,j)\in Q_1$:
\begin{align}
     &w_1,w_4,w_5\in \mathcal{R}^k_{(a,b),(i,j)},\\
     &w_6\in\mathcal{R}^k_{(a,b),(i,j)} \quad\text{if}\quad b\geq \ceil*{\frac{a}{2}}+2,\; 2a+3\leq k\leq a+2b-1\,. \label{claim1b}
\end{align}
\begin{proof}
    \begin{enumerate}[(a)]
         \item \label{proof1a} First we show that $\forall (i,j)\in Q_1$ we have $w_1\in \mathcal{R}^k_{(a,b),(i,j)}$. We use the BMW formula with $\lambda=(a,b)$, $\mu=(i,j)$ and $\nu=w_1=(a-1,b+2)$ and the goal is to prove $N^{(k)w_1}_{(a,b),(i,j)}\geq 1$ which is equivalent to $\operatorname{min}(k,\kmax)\geq \kmin$. Since $(i,j)\in Q_1$ we have $i=2k-3p$ and $j=3p-k$, so we find
        \begin{equation}
            k_0^{\text{max}}=a+b+1+k-p
        \end{equation}
        and
        \begin{equation}
            k_0^{\text{min}}=\operatorname{max}(k,a+b+1+2p-k,a+p,b+1+p)\,.
        \end{equation}
        We notice that $\kmax \geq k$ since $p\leq a+b$. Then it suffices to show 
        \begin{equation}
            \operatorname{min}(k,k_0^{\text{max}})\geq k_0^{\text{min}}\Longleftrightarrow k\geq \kmin \Longleftrightarrow \left\{\begin{array}{lr}
        k \geq a+b+1 \\
        k\geq a+b+1+2p-k \\
        k\geq a+p \\
        k\geq b+1 +p
        \end{array}\right. 
        \end{equation}
    It is clear that the system of four inequalities is satisfied due to the constraints on $k$ and $p$ in the set $Q_1$ \eqref{d1}.
    \item  We will now show that $\forall(i,j)\in Q_1:\; w_4\in  \mathcal{R}^k_{(a,b),(i,j)}$. We proceed as described in (a) to find
    \begin{equation}
        \kmax=a+b+k-p \quad\text{and} \quad \kmin=\operatorname{max}(k,a+b+2p-k,b+1+p,a-1+p)\,.
    \end{equation}
    Since $\kmax \geq k$ it suffices to show:
    \begin{equation}
        k\geq \kmin \Longleftrightarrow \left\{\begin{array}{lr}
        k \geq a+b+2p-k \\
        k\geq b+1+p \\
        k\geq a-1+p 
        \end{array}\right. \,.
    \end{equation}
    The system of three inequalities above is satisfied due to $p\leq k-a-1$ in $Q_1$.
     \item \label{proof1c} We proceed with showing $\forall (i,j)\in Q_1:\; w_5 \in  \mathcal{R}^k_{(a,b),(i,j)}$. We calculate
    \begin{equation}
        \kmax=a+b+1+k-p \quad\text{and}\quad \kmin =\operatorname{max}(a+b+2,k,a+b+1+2p-k,a+1+p)\,.
    \end{equation}
    It suffices to show
    \begin{equation}
        \operatorname{min}(k,\kmax)\geq \kmin \Longleftrightarrow k\geq \kmin \Longleftrightarrow \left\{\begin{array}{lr}
        k \geq a+b+2 \\
        k\geq a+b+1+2p-k \\
        k\geq a+1+p 
        \end{array}\right. \,.
    \end{equation}
The system of   three inequalities is satisfied due to the fact $p\leq k-a-1$ in $Q_1$.
\item \label{proof1d} Finally we will prove \eqref{claim1b}. In this case we find
        \begin{equation}
            \kmax=a+b-1+k-p \quad\text{and}\quad \kmin=\operatorname{max}(k,b+p,a-1+p,a+b-1+2p-k)\,.
        \end{equation}
        It is clear to check that due to the restrictions on $k$ and $p$  we obtain $\kmax\geq k$. Then it suffices to prove 
        \begin{equation}
            k\geq \kmin \Longleftrightarrow  \left\{\begin{array}{lr}
        k \geq b+p \\
        k\geq a-1+p\\
        k\geq a+b-1+2p-k
        \end{array}\right. \,.
        \end{equation}
        The system of three inequalities above is satisfied since $p\leq k-a-1$ in $Q_1$.
    \end{enumerate}
\end{proof}
\item{Statement 2:} For all $(i,j)\in \overline{Q}_1$:
\begin{align}
    &w_1,w_3,w_5\in \mathcal{R}^k_{(a,b),(i,j)},\\
    &w_4\in \mathcal{R}^k_{(a,b),(i,j)} \quad \text{if}\quad b\geq \ceil*{\frac{a+3}{2}},\;k=a+2b, \label{claim2b}\\
   & w_6\in \mathcal{R}^k_{(a,b),(i,j)} \quad \text{if}\quad b\geq \ceil*{\frac{a}{2}}+2,\; 2a+3\leq k\leq a+2b-1\,. \label{claim2c}
\end{align}
\begin{proof}
    \begin{enumerate}[(a)]
         \item We first show that $\forall(i,j)\in \overline{Q}_1$ we have $w_1\in \mathcal{R}^k_{(a,b),(i,j)}$.  Since $(i,j)\in \overline{Q}_1$ we have $i=3p-k,\;j=2k-3p$ and we obtain
    \begin{equation}
        \kmax=a+b+k-p
    \end{equation}
    and
    \begin{equation}
        \kmin=\operatorname{max}(k,b+1+p,a+b+2p-k,a-1+p)
    \end{equation}
    which is less or equal to the value of $\kmin$ in \ref{proof1a} of Statement 1. We also have $k\leq \kmax$ since $p\leq a+b$. Therefore, we immediately can conclude from the calculation in \ref{proof1a}  that $\operatorname{min}(k,\kmax)\geq \kmin \Longleftrightarrow k\geq \kmin$ holds, which proves the statement.
\item We proceed with  proving that $\forall (i,j)\in \overline{Q}_1:\; w_3\in \mathcal{R}^k_{(a,b),(i,j)}$.  In this case we find
    \begin{equation}
        \kmax=a+b+k-p \quad\text{and}\quad \kmin=\operatorname{max}(k,a+b+2p-k,a+1+p)\,.
    \end{equation}
     Since $\kmax \geq k$ it suffices to show:
    \begin{equation}
        k\geq \kmin \Longleftrightarrow \left\{\begin{array}{lr}
        k \geq a+b+2p-k \\
        k\geq a+1+p 
        \end{array}\right. \,.
    \end{equation}
     The system of two inequalities above is satisfied due to $p\leq k-a-1$ in $\overline{Q}_1$.
    \item We now show that $\forall(i,j)\in \overline{Q}_1$ we have $w_5\in \mathcal{R}^k_{(a,b),(i,j)}$. In this case we find
    \begin{equation}
        \kmax=a+b+1+k-p \quad \text{and}\quad \kmin=\operatorname{max}(a+b+2,k,b+1+p,a+p,a+b+1+2p-k)\,.
    \end{equation}
    Notice that the $\kmin$ above is less or equal to the value of $\kmin$ found in \ref{proof1c} of Statement 1  and also the values of $\kmax$ are identical. Therefore it immediately follows that $\operatorname{min}(k,\kmax)\geq \kmin$.
    \item We will now prove \eqref{claim2b}. We find
    \begin{equation}
        \kmax= 2a+3b-p-1 \quad \text{and}\quad \kmin=\operatorname{max}(a+2b,a-1+p,2p-b-1)\,.
    \end{equation}
    By noticing the $p$ interval in $\overline{Q}_1$ and since $k\leq 2a+b$ we conclude that $p\neq a+b$, hence $\kmax \geq k$. Then it suffices to show
    \begin{equation}
            k\geq \kmin \Longleftrightarrow  \left\{\begin{array}{lr}
        a+2b\geq a-1+p\\
        a+2b\geq 2p-b-1
        \end{array}\right. \,.
    \end{equation}
    This system of two inequalities is satisfied because $p\leq k-a-1=2b-1$ in $\overline{Q}_1$ for $k=a+2b$.
    \item Finally we show \eqref{claim2c}. We calculate
        \begin{equation}
            \kmax=a+b-1+k-p \quad\text{and}\quad \kmin=\operatorname{max}(k,a+b-1+2p-k)\,.
        \end{equation}
        We notice that $\kmax$ is the same as the $\kmax$ in \ref{proof1d} of Statement 1 while $\kmin$ is less or equal than the $\kmin$ in \ref{proof1d}. There we immediately conclude that $\operatorname{min}(k,\kmax)\geq \kmin$.
    \end{enumerate}
\end{proof}
\item{Statement 3:} For all $(i,j)\in Q_2$:
\begin{align}
    &w_6\in \mathcal{R}^k_{(a,b),(i,j)}\quad\text{if}\quad a+b+1\leq k\leq a+2b-1,\\
    &w_4\in \mathcal{R}^k_{(a,b),(i,j)} \quad \text{if}\quad a\neq b,\; a+2b\leq k\leq 2a+b,\\
    & w_1\in \mathcal{R}^k_{(a,b),(i,j)} \quad\text{if}\quad a\neq b,\; a+2b+1\leq k\leq 2a+b\,.
\end{align}
\item{Statement 4:} For all $(i,j)\in \overline{Q}_2$:
\begin{align}
    &w_6\in \mathcal{R}^k_{(a,b),(i,j)} \quad \text{if}\quad a+b+1\leq k\leq a+2b-1\\
    & w_1\in \mathcal{R}^k_{(a,b),(i,j)} \quad \text{if}\quad a\neq b,\; a+2b\leq k\leq 2a+b,\\
    & w_4\in  \mathcal{R}^k_{(a,b),(i,j)} \quad \text{if}\quad a\neq b,\;k=a+2b\,.
\end{align}
\item{Statement 5:} For all $(i,j)\in Q_3$:
\begin{equation}
    w_2,w_3,w_4,w_6\in  \mathcal{R}^k_{(a,b),(i,j)}\,.
\end{equation}
\item{Statement 6:} For all $(i,j)\in \overline{Q}_3$:
\begin{equation}
    w_1,w_4,w_6\in  \mathcal{R}^k_{(a,b),(i,j)}\,.
\end{equation}
\item{Statement 7:} For all $(i,i)\in Q_4$:
\begin{equation}
    w_6\in  \mathcal{R}^{a+b}_{(a,b),(i,i)}\,.
\end{equation}
\item{Statement 8:} For all $(i,i)\in Q_5$:
\begin{align}
    &w_1 \in  \mathcal{R}^k_{(a,b),(i,i)} \\
    & w_4\in  \mathcal{R}^k_{(a,b),(i,i)} \quad \text{if} \quad b\neq0,\; b\leq a-2 \\
    & w_6 \in  \mathcal{R}^k_{(a,b),(i,i)} \quad \text{if} \quad b\neq 0
\end{align}
\item{Statement 9:} For all $(i,i)\in Q_6$:
\begin{align}
    &w_1,w_2,w_5\in  \mathcal{R}^k_{(a,b),(i,i)},\\
    & w_6 \in  \mathcal{R}^k_{(a,b),(i,i)} \quad\text{if}\quad b\geq \ceil*{\frac{a+3}{2}},\;2a+2\leq k\leq a+2b-1,\\
    & w_6\in  \mathcal{R}^k_{(a,b),(i,i)} \quad\text{if}\quad b\geq \ceil*{\frac{a}{2}}+1,\;a\neq b,\; k=a+2b,\, a \;\text{is even}\,.
\end{align}
\item{Statement 10:} For all $(i,j)\in C_1\cup \overline{C}_1$:
\begin{equation}
    w_1,w_5\in  \mathcal{R}^k_{(a,b),(i,j)}\,.
\end{equation}
\item{Statement 11:} For all $(i,j)\in C_2$:
\begin{equation}
    w_1,w_5\in  \mathcal{R}^k_{(a,b),(i,j)}\,.
\end{equation}
\item{Statement 12:} For all $(i,j)\in \overline{C}_2$:
\begin{align}
     &w_5\in  \mathcal{R}^k_{(a,b),(i,j)},\\
     & w_1\in  \mathcal{R}^k_{(a,b),(i,j)} \quad\text{if}\quad a\neq b\,.
\end{align}
\item{Statement 13:} For all $(i,j)\in C_3\cup \overline{C}_3$:
\begin{align}
   & w_1,w_4,w_5,w_6\in  \mathcal{R}^k_{(a,b),(i,j)}\,.
\end{align}
\item{Statement 14:} For all $(i,i)\in C_5$:
\begin{equation}
    w_1\in  \mathcal{R}^k_{(a,b),(i,i)}\,.
\end{equation}
\item{Statement 15:} For all $(i,i)\in C_6$:
\begin{equation}
    w_1,w_5 \in  \mathcal{R}^k_{(a,b),(i,i)}\,.
\end{equation}

\end{description}

 \bibliographystyle{JHEP}
 %\bibliography{biblio.bib}

\begin{thebibliography}{99}

\bibitem{Buican} M. Buican, L. Li, R. Radhakrishnan
\emph{$a \times b = c$ in $2+1$D TQFT},
Quantum {\bf 5}, 468 (2021); 		arXiv:2012.14689.


\bibitem{Graham_Watts} K. Graham and G. M. T. Watts, 
\emph{Defect Lines and Boundary Flows},
JHEP {\bf 04} (2004) 019; arXiv:hep-th/0306167.

\bibitem{5A}
C.-M. Chang, Y.-H. Lin, S.-H. Shao, Y. Wang and X. Yin,
\emph{Topological defect lines and renormalization group flows in two dimensions},
Journal of High Energy Physics {\bf 2019} (2019) 26; arXiv:1802.04445.



\bibitem{Sakura_etal1}  L. Bhardwaj, L. E. Bottini, D. Pajer, S. Schafer-Nameki,
\emph{Gapped Phases with Non-Invertible Symmetries: (1+1)d},
arXiv:2310.03784.

\bibitem{Sakura_etal2} L. Bhardwaj, L. E. Bottini, D. Pajer, S. Schafer-Nameki,
\emph{Categorical Landau Paradigm for Gapped Phases}, arXiv:2310.03786.

\bibitem{FGS} S. Fredenhagen, M. R. Gaberdiel, and C. Schmidt-Colinet, 
\emph{Bulk flows in Virasoro
minimal models with boundaries} , J. Phys. {\bf A42} 49, (2009) 495403; arXiv:0907.2560.

\bibitem{Gaiotto} D. Gaiotto, \emph{Domain Walls for Two-Dimensional Renormalization Group Flows}, JHEP
{\bf 12} (2012)103; arXiv:1201.0767.

\bibitem{Kom_etal} Z. Komargodski, K. Ohmori, K. Roumpedakis, and S. Seifnashri, \emph{Symmetries and
strings of adjoint QCD2}, JHEP {\bf 03} (2021)103; arXiv:2008.07567.

\bibitem{Cordova_etal} C. Cordova, D. García-Sepúlveda, and N. Holfester, \emph{Particle-soliton degeneracies from
spontaneously broken non-invertible symmetry}, JHEP {\bf 07} (2024)154;
arXiv:2403.08883.

\bibitem{Kon1}
A. Konechny,
\emph{Open topological defects and boundary {RG} flows},
J. Phys. {\bf A53}  (2020) 155401; 	arXiv:1911.06041.

\bibitem{TanakaNakayama} T. Tanaka and Yu. Nakayama, \emph{Infinitely many new renormalization group flows between Virasoro minimal models from non-invertible symmetries}, arXiv:2407.21353

\bibitem{FFRS}
J. Fröhlich, J. Fuchs, I. Runkel and C. Schweigert,
\emph{Duality and defects in rational conformal field theory},
Nuclear Physics B {\bf 763} (2007) 354; 	arXiv:hep-th/0607247.

\bibitem{Runkelcharges} I. Runkel,
\emph{Non-local conserved charges from defects in perturbed conformal field theory},
Journal of Physics A: Mathematical and Theoretical {\bf 43}, 365206 (2010); arXiv:1004.1909 .


\bibitem{KW} J. Fr\" ohlich, J. Fuchs, I. Runkel and C. Schweigert, 
\emph{Kramers-Wannier duality from conformal defects}, 
Phys. Rev. Lett. {\bf 93} (2004) 070601; arXiv:cond-mat/0404051.

\bibitem{PZ}
V.B. Petkova and J.-B. Zuber,
\emph{Generalised twisted partition functions},
Physics Letters {\bf B504} (2001) 157; arXiv:hep-th/0011021.

\bibitem{FRS1}
J. Fuchs, I. Runkel and C. Schweigert,
\emph{{TFT} construction of {RCFT} correlators {I}: partition functions},
Nuclear Physics B {\bf 646} (2002) 353; 	arXiv:hep-th/0204148.

\bibitem{FRS2}
J. Fuchs, I. Runkel and C. Schweigert,
\emph{{TFT} construction of {RCFT} correlators {II}: unoriented world sheets},
Nuclear Physics B {\bf 678} (2004) 511; 	arXiv:hep-th/0306164.

\bibitem{FRS3}
J. Fuchs, I. Runkel and C. Schweigert,
\emph{{TFT} construction of {RCFT} correlators {III}: Simple currents},
Nuclear Physics B {\bf 694} (2004) 277; arXiv:hep-th/0403157.

\bibitem{FRS4}
J. Fuchs, I. Runkel and C. Schweigert,
\emph{{TFT} construction of {RCFT} correlators {IV}: Structure constants and correlation functions},
Nuclear Physics B {\bf 715} (2005) 539; 	arXiv:hep-th/0412290.

\bibitem{FRS5}
J. Fjelstad, J. Fuchs, I. Runkel and C. Schweigert,
\emph{{TFT} construction of {RCFT} correlators {V}: Proof of modular invariance and factorisation},
Theor. Appl. Categor. {\bf 16} (2006) 342; 	arXiv:hep-th/0503194.


\bibitem{WZWtadpole} A. Urichuk and M. A. Walton, \emph{Adjoint aﬃne fusion and tadpoles}, J.
Math. Phys., {\bf 57(6)}, 061702 (2016).

\bibitem{paper1} Konechny, A., Vergioglou, V. \emph{On fusing matrices associated with conformal boundary  conditions}, J. High Energ. Phys.,  {\bf 2024}, 142 (2024); 	arXiv:2405.10189.





\bibitem{WZW1} E. Witten \emph{Non-abelian bosonization in two dimensions},
Commun.Math. Phys. \textbf{92}, 455–472 (1984).

\bibitem{WZW2} V. Knizhnik, A. Zamolodchikov, \emph{Current algebra and Wess-Zumino model in two dimensions}, Nucl. Phys. B \textbf{247}, 83-103 (1984).

\bibitem{WZW6} D. Gepner, E. Witten, \emph{String theory on group manifolds}, Nucl. Phys. B \textbf{278}, 493 (1986).

%\bibitem{WZW3} Y.-Z. Huang and J. Lepowsky, \emph{Intertwining operator algebras and vertex tensor categories for affine Lie algebras}, Duke Math. J. \textbf{99}, 113–134 (1999).

%\bibitem{WZW4}Y.-Z. Huang, \emph{Rigidity and modularity of vertex tensor categories}, Comm. Contemp. Math. \textbf{10}, 871–911 (2008).

%\bibitem{WZW5} Y.-Z. Huang, J. Lepowsky, \emph{Tensor categories and the mathematics of rational and logarithmic conformal field theory}, Journal of Physics A: Mathematical and Theoretical \textbf{46}, (2013); 	arXiv:1304.7556.


\bibitem{Schellekens01} A. Schellekens, S. Yankielowicz
\emph{Extended chiral algebras and modular invariant partition functions},
INuclear Physics B \textbf{327} 3, 673-703 (1989).

\bibitem{Schellekens02} A. Schellekens, S. Yankielowicz
\emph{Simple Currents, Modular Invariants and Fixed Points},
International Journal of Modern Physics A {\bf 5}, 2903-2952 (1990).



\bibitem{Fuchs01} J. Fuchs, 
\emph{Simple WZW currents},
Commun.Math. Phys {\bf 136}, 345–356 (1991).


\bibitem{BMW} L. Begin, P. Mathieu, M.A. Walton, 
\emph{$\widehat{{\rm su}
}\left( 3 \right)_k $ Fusion Coefficients},
Modern Physics Letters A {\bf 07}, No. 35, 3255–3265 (1992);	arXiv:hep-th/9206032.

\bibitem{Barker} A. Barker, D. Swinarski, J. Wu, L. Vogelstein,
\emph{A new proof of a formula for the type $A_2$ fusion rules},
J. Math. Phys. {\bf 56}, 011703 (2005); 	arXiv:1408.4353.



\bibitem{Bartsch_Bullimore} T. Bartsch,  M. Bullimore, and A. Grigoletto, \emph{Representation theory for categorical symmetries
}, arXiv:2305.17165.



\end{thebibliography}

%% or
%% [B] Manual formatting (see below)
%% (i) We suggest to always provide author, title and journal data or doi:
%% in short all the informations that clearly identify a document.
%% (ii) please avoid comments such as "For a review'', "For some examples",
%% "and references therein" or move them in the text. In general, please leave only references in the bibliography and move all
%% accessory text in footnotes.
%% (iii) Also, please have only one work for each \bibitem.

\end{document}